\documentclass[mnsc,sglanonrev]{informs4_clean}
\usepackage{eqndefns-left} 
\RequirePackage{tgtermes}
\RequirePackage{newtxtext}
\RequirePackage{newtxmath}
\RequirePackage{bm}
\RequirePackage{endnotes}

\OneAndAHalfSpacedXII 

\usepackage{algorithm}
\usepackage{algpseudocode}
\usepackage{tikz}
\usepackage{mathtools}
\usepackage{amsmath}
\usepackage{graphicx}
\usepackage{xcolor}
\usepackage{svg}
\usepackage{makecell}
\usepackage{booktabs}
\usepackage{multirow}
\usepackage{tabularx}
\usepackage{pdflscape} 
\usepackage{pdflscape}      
\usepackage{caption}       
\usepackage{adjustbox}

\usepackage{tikz}
\usetikzlibrary{shapes.geometric, arrows.meta, positioning}

\tikzstyle{startstop} = [rectangle, rounded corners, minimum width=4.5cm, minimum height=1cm,text centered, draw=black, fill=gray!10]
\tikzstyle{decision} = [
  diamond, aspect=2.2, draw=black, fill=blue!10,
  text centered, inner sep=1pt,
  minimum width=0.5\textwidth,
  text width=0.45\textwidth,
  align=center
]
\tikzstyle{process} = [rectangle, draw=black, fill=green!10, text centered, minimum width=4.5cm, minimum height=1cm]
\tikzstyle{arrow} = [thick, ->, >=stealth]

\newcommand{\pref}[1]{(\ref{#1})}
\newcommand{\fref}[1]{Figure \ref{#1}}
\newcommand{\tref}[1]{Table \ref{#1}}

\newcommand{\sref}[1]{Section \ref{#1}}
\newcommand{\ssref}[1]{Section \ref{#1}}

\newcommand{\proman}[1]{(\romannumeral #1)}

\newcommand{\edit}[1]{\textcolor{black}{#1}}

\usepackage[colorlinks=false]{hyperref}
\usepackage{cleveref}

\usepackage{natbib}
 \bibpunct[, ]{(}{)}{,}{a}{}{,}%
 \def\bibfont{\small}%
\EquationsNumberedThrough    

\TheoremsNumberedThrough     
\ECRepeatTheorems  %

\MANUSCRIPTNO{MNSC-0001-2024.00}

\allowdisplaybreaks[4]
\begin{document}


\RUNAUTHOR{Anonymous Authors}

\RUNTITLE{Statistical Arbitrage in Options Markets}


\TITLE{Statistical Arbitrage in Options Markets by Graph Learning and Synthetic Long Positions}


\ARTICLEAUTHORS{%
\AUTHOR{Yoonsik Hong}
\AFF{Department of Industrial Engineering and Management Sciences,
Northwestern University, \EMAIL{YoonsikHong2028@u.northwestern.edu}}

\AUTHOR{Diego Klabjan}
\AFF{Department of Industrial Engineering and Management Sciences,
Northwestern University, \EMAIL{d-klabjan@northwestern.edu}}
%
%
} 

\ABSTRACT{%
Statistical arbitrages (StatArbs) driven by machine learning has garnered considerable attention in both academia and industry. 
Nevertheless, deep-learning (DL) approaches to directly exploit StatArbs in options markets remain largely unexplored. 
Moreover, prior graph learning (GL)—a methodological basis of this paper—studies overlooked that features are tabular in many cases and that tree-based methods outperform DL on numerous tabular datasets.  
To bridge these gaps, we propose a two-stage GL approach for direct identification and exploitation of StatArbs in options markets. 
In the first stage, we define a novel prediction target isolating pure arbitrages via synthetic bonds. 
To predict the target, we develop RNConv, a GL architecture incorporating a tree structure. 
In the second stage, we propose SLSA—a class of positions comprising pure arbitrage opportunities. 
It is provably of minimal risk and neutral to all Black-Scholes risk factors under the arbitrage-free assumption. 
We also present the SLSA projection converting predictions into SLSA positions. 
Our experiments on KOSPI 200 index options show that RNConv statistically significantly outperforms GL baselines, and that SLSA consistently yields positive returns, achieving an average P\&L-contract information ratio of 0.1627. 
Our approach offers a novel perspective on the prediction target and strategy for exploiting StatArbs in options markets through the lens of DL, in conjunction with a pioneering tree-based GL.
}


\maketitle

\section{Introduction}
\label{sec:introduction}


We address the problem of directly identifying and exploiting statistical arbitrage (StatArb) opportunities in options markets by graph learning (GL). 
The concepts of StatArb strategies, though defined differently across studies, can be summarized as a class of trading strategies that generate a positive expected return with an acceptable risk of loss by leveraging price discrepancies \citep{lazzarino2018statistical}. 
StatArbs have attracted considerable attention in both academia and the financial industry, where they have been widely researched and applied \citep{zhan2022exploring}. 
To capture and exploit StatArbs in options markets, in this paper, we develop a GL method and trading strategy. 


Although various machine learning approaches have been proposed in the options markets, they did not focus on direct identification of StatArbs.  
Specifically, \cite{zapart2003statistical} treated StatArbs as an ancillary subject, and although their strategy was labeled as such, it remained substantially exposed to the canonical risk factors of~\cite{black1973pricing}.
\cite{horikawa2024relationship, franccois2025difference} employed StatArbs solely as a lens through which to compare deep hedging and delta hedging strategies, rather than as a phenomenon to be directly uncovered. 
Other machine learning studies in the options markets likewise fall short of addressing the direct identification of StatArbs. 
Consequently, the direct detection of StatArbs via machine learning remains largely uncharted in the literature. 
Our approach directly identifies and capitalizes on StatArbs, while maintaining a prudently low level of risk.


Prior GL methods, a core methodological basis of this study, overlooked that features are tabular in many cases. 
GL excels at learning structured representations through node and edge embeddings \citep{bhatti2023deep}. 
Options market data are inherently  tabular but exhibit relational structure. 
For instance, the maturity and strike prices are tabular features, and options sharing the same underlying instrument are related. 
Treating such features as node or edge features enables the straightforward use of GL; however, tabular data are not embeddings.
Although deep learning can embed such features, it underperforms compared to tree-based models on numerous tabular datasets \citep{shwartz2022tabular,borisov2022deep,mcelfresh2023neural,rana2023comparative,yildiz2024gradient}. 
Hence, rather than relying on deep-learning-based embeddings, 
tree-based embeddings within GL architectures can be a more promising avenue,
which is realized in our method based on tree-based embeddings. 

We present a two-step approach bridging the two key gaps related to the subject and architectures: the lack of StatArb trading strategies in options markets and the insufficient consideration of the tabular features in GL.
The first step focuses on developing a graph nueral network (GNN) architecture that incorporates a tree-based architecture to detect StatArb opportunities, while the second step aims to exploit the identified opportunities through an appropriate, new trading strategy.

In the first step, we develop a GNN architecture, \textbf{R}evised \textbf{N}eural oblivious decision ensemble graph \textbf{Conv}olution (RNConv), along with a prediction target that exclusively contains arbitrages. 
Distinct from prior studies, 
our prediction target is a provably pure arbitrage opportunity, arising from deviations in the prices of synthetic zero-coupon bonds constructed from synthetic-long and underlying instrument positions.
These prices should be identical under the law of one price---or more strictly, under the \textbf{a}rbitrage-\textbf{f}ree (AF) assumption---but diverge in practice (see \fref{fig:arb_time_series}) due to market inefficiencies.
Next, we construct graphs where a node is a put-call pair sharing the same maturity and strike price, with node features comprising past arbitrage opportunities, implied volatility, moneyness, and time to maturity. 
Then, RNConv extracts market-level and idiosyncratic information from the graph and each node, respectively, to predict arbitrage opportunities.
Our RNConv architecture incorporates a tree structure built upon \textbf{N}eural \textbf{O}blivious \textbf{D}ecision \textbf{E}nsembles (NODE)~\citep{popov2019neural}, 
yielding superior performance compared to other GNN architectures on the KOSPI 200 index option datasets. 
RNConv is a \edit{single} model combining a standard GNN with differentiable trees and is trained as a single unit.

To exploit StatArbs in the second step, we present \textbf{S}ynthetic-\textbf{L}ong-\textbf{S}hort-\textbf{A}rbitrage (SLSA) and an SLSA projection method.
We define SLSA as a class of synthetic-long positions satisfying certain constraints.
Every SLSA position exclusively contains arbitrage opportunities, which is confirmed by the property that the price of SLSA positions vanishes under the AF assumption. 
Moreover, SLSA values exhibit zero variance and are neutral to all \cite{black1973pricing} risk factors under the AF assumption. 
Every SLSA position carries minimal risk under the assumption of an AF market or the existence of a present value function, both of which are common.
Although real-world markets may contain arbitrages, they are expected to become AF \citep{varian1987arbitrage,langenohl2018sources}, thereby enabling SLSA to have low risk.
Lastly, we present the SLSA projection method which converts predictions into an SLSA position that yields StatArbs.

From our two-step approach, we set forth a novel formulation of the prediction target and strategy for exploiting StatArbs in options markets through the lens of deep learning in conjunction with a pioneering tree-based GNN.
The main contributions of this study are as follows. 
\begin{itemize}
    \item We present a novel machine-learning prediction target comprising pure arbitrages in options markets.
    \item To the best of our knowledge, this work is the first to introduce RNConv, a graph convolutional architecture that integrates neural trees in order to effectively exploit tabular node features while capturing relational structures inherent in option markets. 
    \item To the best of our knowledge, we are the first to propose SLSA, a class of synthetic-long positions that solely contain arbitrage opportunities and are theoretically low-risk and neutral to all \cite{black1973pricing} risk factors under the arbitrage-free assumption, and present the SLSA projection method that transforms arbitrage predictions into such SLSA positions capable of generating StatArbs.
    \item We demonstrate the effectiveness of our approach on the KOSPI 200 index options, where RNConv statistically significantly outperform benchmarks and SLSA achieves steadily increasing P\&Ls with an average P\&L-contract information ratio of 0.1627.
\end{itemize}

This paper is structured as follows.  
\sref{sec:literature} reviews literature. 
\sref{sec:preliminaries} introduces the preliminaries and background.  
\sref{sec:setup} explains the overarching setup and notation. 
\sref{sec:method} presents our proposed methods, including the prediction target, RNConv, and SLSA. 
\sref{sec:experiments} reports experimental results in the KOSPI 200 index options market. 
Finally, \sref{sec:conclusion} concludes the paper.

\section{Related Work}
\label{sec:literature}
Our study fills gaps in the existing literature concerning learning-based methodologies for StatArbs in option markets, as well as the embedding of tabular features within GNN architectures.
In this section, we review relevant prior work with respect to these two aspects, highlight their limitations, and articulate the motivation behind our approach.

\subsection{Learning-Based Approaches in Option Markets}

Since strong assumptions are required in mathematical models (e.g., geometric Brownian motion in \citep{black1973pricing}, jump diffusion process in \citep{merton1976option}, displaced diffusion process in \citep{rubinstein1983displaced}), neural network approaches to problems in option markets have drawn attention \citep{malliaris1993neural,anders1998improving,wang2024considering,farahani2024black}. 
Neural networks have been studied in a wide range of problems in option markets such as option pricing \citep{malliaris1993neural,dugas2009incorporating,wang2024considering}, implied volatility estimation \citep{malliaris1996using,osterrieder2020neural,zheng2021incorporating}, hedging strategies \citep{shin2012dynamic,chen2012pricing,buehler2019deep}, the optimal stopping problem \citep{kohler2010pricing, becker2019deep}, and risk-neutral density estimation \citep{schittenkopf2001risk}.
Moreover, various international options markets have been covered employing neural networks, including the KOSPI 200 in South Korea \citep{choi2004efficient,park2014parametric}, the S\&P 500 in the United States \citep{yang2017gated,zheng2021incorporating}, the FTSE 100 in the United Kingdom \citep{zhang2021option}, and the DAX 30 in Germany \citep{liu2019wavelet}.

Although deep learning has been applied to a wide range of problems in options markets, the predominant focus in the literature has been on option pricing \citep{ruf2019neural}, while StatArb identification and exploitation in options markets has been addressed only as a secondary topic. 
For instance, \cite{zapart2003statistical} used neural networks to predict volatility 
and subsequently priced options, but their model was not designed to identify arbitrage opportunities.
Although \cite{zapart2003statistical} sought to exploit pricing discrepancies through delta-hedged positions, their exposures to volatility, time, interest rates, and higher-order sensitivities to the underlying instrument indicate that they cannot be regarded as StatArbs. 
\cite{horikawa2024relationship,franccois2025difference} analyzed the difference between deep hedging and delta hedging in the context of StatArbs,  
but their main research focus was not on predicting and exploiting StatArbs, which is our topic.

Learning–based StatArb research focused on stock markets (e.g., \cite{
guijarro2021deep,zhao2022deep}) \edit{with a few discussing options}, while no studies in options markets aim to mainly address StatArbs. 
To bridge this gap, we explore StatArbs in an options market employing graph learning.

\subsection{Graph Neural Networks (GNNs)}

GNNs are deep-learning methods for graph-structured data, where nodes represent entities and edges denote their relationships \citep{liu2022introduction}. 
Unlike convolutional neural networks designed for grid-structured data (e.g., images), GNNs capture intricate relational dependencies in non-Euclidean domains \citep{bhatti2023deep}, such as social networks \citep{chen2022gc}. 
Recently, 
\citet{wang2024considering} employed a GNN to capture momentum spillover effects, where an asset’s past return predicts those of related assets \citep{ali2020shared}, between options, and their GNN outperformed canonical neural networks. 
Motivated by their work, we tackle our StatArb problem by developing 
a GNN to extract market signals in the options market.

Various graph convolutions have been presented thus far \citep{wu2020comprehensive}. 
Our paper focuses on four representative GNN architectures—GCN~\citep{kipf2017semisupervised}, SAGE~\citep{hamilton2017inductive}, GAT~\citep{velickovic2018graph}, and GPS~\citep{rampavsek2022recipe}—chosen for their historical significance, architectural diversity, and practical relevance.  
These methods are listed in PyTorch-Geometric \citep{fey2019fast}, which attests to their wide recognition and validation. 
In PyTorch-Geometric~\citep{fey2019fast}, GCN, SAGE, and GAT are the most cited methods while GPS is one of the most recent methods.  
Accordingly, we adopt them as benchmarks to assess the effectiveness of our RNConv.


GCN~\citep{kipf2017semisupervised} is a foundational GNN that utilizes a linear approximation of spectral graph convolutions on graph-structured data.  
SAGE~\citep{hamilton2017inductive} is an inductive framework that generates node embeddings by aggregating information from a node’s local neighborhood.
GAT~\citep{velickovic2018graph} employs an attention mechanism between each node and its neighbors to assign different importances to each neighbor during aggregation. 
GPS~\citep{rampavsek2022recipe} combines global transformer-based attention and local message passing with positional and structural encodings.
This hybrid architecture can extract both local and global information, showing state-of-the-art performance in their experiments.  
Together, these models mark the evolution toward more expressive and scalable graph representation learning.


However, these four methods—--and prior GNN techniques, to our knowledge—--overlook a key property in real-world settings: features are tabular in many cases.   
Tree-based models have been shown to outperform neural networks on numerous tabular datasets~\citep{shwartz2022tabular,borisov2022deep,mcelfresh2023neural,rana2023comparative,yildiz2024gradient}.  
In option pricing, \citet{ivașcu2021option} found that tree-based methods (e.g., LightGBM~\citep{ke2017lightgbm}, XGBoost~\citep{chen2016xgboost}) surpass both standard neural networks and classical parametric models \citep{black1973pricing,corrado1996skewness}. 
Moreover, NODE~\citep{popov2019neural}, an ensemble of differentiable neural trees, outperformed methods based on non-differentiable trees such as CatBoost~\citep{prokhorenkova2018catboost} and XGBoost~\citep{chen2016xgboost} on some tabular datasets.
Accordingly, we propose RNConv, a GNN that incorporates NODE~\citep{popov2019neural} to more effectively leverage tabular node features.

\section{Preliminaries}
\label{sec:preliminaries}
This section discusses NODE \citep{popov2019neural} and the Low-Rank Cross Network \citep{wang2021dcn}, which underpin the basis of our proposed method, RNConv.

\subsection{Neural Oblivious Decision Ensembles (NODE)}
\label{subsec:node}

NODE \citep{popov2019neural} integrates a decision tree ensemble into deep learning.
It reformulates the oblivious decision trees (ODTs) \citep{kohavi1994bottom,lou2017bdt} into differentiable ODTs, assembles them into a NODE layer, and stacks NODE layers.

An ODT is a binary decision tree where all nodes at the same depth share the same splitting criterion \citep{kohavi1994bottom,lou2017bdt}. 
Given an input $\textbf{x}\in\mathbb{R}^p$, the output of a depth-$d$ ODT is defined as \cite{kohavi1994bottom,lou2017bdt}:
\begin{equation}
    h^{ODT}(\textbf{x})=R_{2-\mathbb{I}(\textbf{s}_1^T \textbf{x}-b_1),~...,~2-\mathbb{I}(\textbf{s}_d^T \textbf{x}-b_d)}\in\mathbb{R} \label{eq:odt}
\end{equation}
where $R \in \mathbb{R}^{\overbrace{2 \times \dots \times 2}^{d}}$ is a response tensor; $R_{i_1, i_2,..,i_d}\in\mathbb{R}$ indicates the entry of $R$ at $(i_1, i_2,..., i_d)$; and $\mathbb{I}(\cdot)$ is a Heaviside step function. 
For all depth-$i$ nodes, $\textbf{s}_i\in\{0,1\}^p$ is a one-hot vector (i.e., $\textbf{1}^T \textbf{s}_i=1$, where $\textbf{1}$ denotes the column vector of all ones used consistently throughout this paper) that selects the splitting feature, and $b_i\in\mathbb{R}$ is the threshold. 

\cite{popov2019neural} define a \textbf{d}ifferentiable oblivious \textbf{d}ecision \textbf{t}ree (DDT) $h^{DDT}(\cdot)$ by replacing $\textbf{s}_i^T \textbf{x}$ and $\mathbb{I}(x')$ in \pref{eq:odt} with $\textbf{x}^T entmax_{\alpha}(\hat{\textbf{s}}_{i})$ and a two-class entmax $\sigma_{\alpha}(x')$, respectively:
\begin{equation}
    h^{DDT}(\textbf{x})=\sum_{i_1,...,i_d\in\{1,2\}^d} R_{i_1,...,i_d} \cdot C_{i_1,...,i_d}(\textbf{x}) \label{eq:dodt1}
\end{equation}
where
\begin{align}
    &C(\textbf{x})=\begin{bmatrix} c_1(\textbf{x}) \\ 1 - c_1(\textbf{x}) \end{bmatrix} \otimes ... \otimes \begin{bmatrix} c_d(\textbf{x}) \\ 1 - c_d(\textbf{x}) \end{bmatrix}, \label{eq:dodt2}\\
    &c_i(\textbf{x})=\sigma_{\alpha}\left(\frac{\textbf{x}^T entmax_{\alpha}(\hat{\textbf{s}}_{i})-b_i}{\kappa_i}\right), \label{eq:dodt3}\\
    &\sigma_{\alpha}(x')=entmax_{\alpha}([x'~0]^T). \label{eq:dodt4}
\end{align}
The $\alpha$-entmax transformation, $entmax_\alpha$~\citep{peters2019sparse}, is a sparsity-inducing activation function that generalizes softmax and sparsemax. 
\cite{popov2019neural} set $\alpha$ to 1.5. 
The parameters $R \in \mathbb{R}^{\overbrace{2 \times \dots \times 2}^{d}}, \hat{s}_i \in \mathbb{R}^{p}, b_i \in \mathbb{R}, \kappa_i \in \mathbb{R}$ are learnable where $\kappa_i$ is introduced for scaling.
The operator $\otimes$ denotes the outer product and constructs the response selection tensor $C(\textbf{x}) \in \mathbb{R}^{\overbrace{2 \times \dots \times 2}^{d}}$.
For example, if $c_1(\textbf{x})=1$ (i.e., $\textbf{s}_1^T \textbf{x}-b_1>0$), then $C_{2,i_2,\dots,i_d}(\textbf{x})=0$, while $C_{1,i_2,\dots,i_d}(\textbf{x})$ can be nonzero, allowing only entries of the form $R_{1,i_2,\dots,i_d}$ to be considered.  
Likewise, if $c_1(\textbf{x}) = 1$ and $c_2(\textbf{x}) = 0$, the output is further narrowed to $R_{1,2,\dots,i_d}$. 
Repeating this process ultimately yields a single selected element from $R$.



\citet{popov2019neural} define a NODE layer $g^{NL}(\textbf{x}) \in \mathbb{R}^{1 \times n_{ddt}}$ as a stack of DDTs $h^{DDT}_j$ for $j \in \{1, \dots, n_{ddt}\}$. Formally,
\begin{equation}
    g^{NL}(\textbf{x}) = \left[h^{DDT}_1(\textbf{x})~h^{DDT}_2(\textbf{x})~\dots~h^{DDT}_{n_{ddt}}(\textbf{x}) \right].
\end{equation}

Finally, NODE is defined with NODE layers $g^{NL}_l$, where $l \in \{1, \dots, l_{NODE}\}$, following DenseNet~\citep{huang2017densely}.  
Its output is given by:
\begin{equation}
    f^{NODE} (\textbf{x}) = \sum_{l=1}^{l_{NODE}} {g_l^{NL} (\textbf{z}_l) \textbf{1}} \quad  \in \mathbb{R} \label{eq:node} \\
\end{equation}
where
\begin{equation}
    \textbf{z}_l = \begin{cases}
    \textbf{x} & \text{if}~l=1 \\
    \left[\textbf{x}^T~g_1^{NL}(\textbf{z}_1)~...~ g_{l-1}^{NL}(\textbf{z}_{l-1})\right]^T     & \text{if}~l>1,
    \end{cases} \label{eq:node1}
\end{equation}
and $\textbf{z}_l\in\mathbb{R}^{p+(l-1)n_{odt}}$. 
Each NODE layer receives the original input and all previous layer outputs, and the final output is the sum of all layer outputs.

\subsection{Low-Rank Cross Network}
\label{subsec:dcn-v2}
The Low-Rank Cross Network \cite{wang2021dcn}, addresses the limitation observed by \citet{beutel2018latent, wang2017deep} that deep neural networks are inefficient at learning feature interactions.
\citet{wang2021dcn} define the ($l+1$)-th layer of a cross network given an input $\textbf{x}=\textbf{x}_0\in\mathbb{R}^p$ as:
\begin{equation}
\textbf{x}_{l+1} = \textbf{x} \odot \left(W_l\textbf{x}_l+\textbf{b}_l\right) +\textbf{x}_l    \label{eq:original_dcn0}
\end{equation}
where $W_l\in\mathbb{R}^{p \times p}$ is learnable, and $\odot$ is element-wise multiplication.
Since $\textbf{x}$ and $\textbf{x}_l$ are multiplied in \eqref{eq:original_dcn0}, the network captures feature interactions. 
By applying a low-rank technique, \citet{wang2021dcn} derived the Low-Rank Cross Network, whose $(l+1)$-th layer is given by:
\begin{equation}
    \textbf{x}_{l+1} = \textbf{x}_0 \odot \left(U_l^{CN}\cdot \phi(C_l^{CN} \cdot \phi({V_l^{CN}}^T \textbf{x}_l))+\textbf{b}_l\right) +\textbf{x}_l \label{eq:original_dcn}
\end{equation}
where $\phi$ is an activation function, and $U_l^{CN},V_l^{CN}\in\mathbb{R}^{p\times p_{CN}}$, $C_l^{CN}\in\mathbb{R}^{p_{CN}\times p_{CL}}$, $\textbf{b}_l\in\mathbb{R}^{p}$ are learnable.

\section{Setup and Notation}
\label{sec:setup}
We delineate the overarching setup adopted throughout the study, encompassing the investment environment, asset types, notation, prediction dataset structure, and underlying assumptions.

\subsection{Investment Environment}
\label{subsec:setup_env}

\textbf{Time Horizon}.
We assume there exists a bijective mapping from all trading dates to $\mathbb{N}$ that preserves temporal order. 
We map each time point on each trading date to $\mathcal{T} = [1, \infty)$ by assigning time point $HH$:$MM$:$SS$ on the $t$-th trading date ($t \in \mathbb{N}$) to $t + \frac{HH \times 60^2 + MM \times 60 + SS}{24 \times 60^2} \in \mathcal{T}$.
Thus, $t \in \mathbb{N}$ and $\tau \in \mathcal{T}$ denote a trading date and a time point, respectively.
We use $\lfloor \tau \rfloor \in \mathbb{N}$ to refer to the trading date corresponding to $\tau \in \mathcal{T}$, where $\lfloor \cdot \rfloor$ denotes the floor function.
We define $o(t)$ as a function that maps each $t \in \mathbb{N}$ to a value in $[t,t+1)\subset\mathcal{T}$ corresponding to the market open time point on the trading date $t$. 
For example, if the market opens at 08:30:00 on trading date $t$, then $o(t) = t + \frac{8 \times 60^2 + 30 \times 60}{24 \times 60^2}$.

\textbf{Assets}.  
An asset $a$ is defined as a structured tuple $(c; p)$, where the semicolon separates the asset type $c$ from the parameter tuple $p$. 
For example, $(PT; M, K)$ is a put option $PT$ with maturity $M$ and strike price $K$, where $c = PT$ and $p = (M, K)$.
In this paper, we only consider asset types from the set $\mathcal{C} = \{UI, PT, CL, SL, LS, SA, RF\}$, where:
\begin{itemize}
    \item $UI$: \textbf{u}nderlying \textbf{i}nstrument,
    \item $PT$: European \textbf{p}u\textbf{t} option,
    \item $CL$: European \textbf{c}al\textbf{l} option,
    \item $SL$: \textbf{s}ynthetic \textbf{l}ong \edit{(see below for details)},
    \item $LS$: synthetic \textbf{l}ong-\textbf{s}hort, as defined in Section \ref{subsec:slsa},
    \item $SA$: synthetic long-\textbf{s}hort \textbf{a}rbitrage, as defined in Section \ref{subsec:slsa}, and
    \item $RF$: \textbf{r}isk-\textbf{f}ree asset.
\end{itemize}

For $c \in \mathcal{C}$, let $\mathcal{A}(c) = \{ a \mid a = (c; p) \}$ denote the set of all assets of type $c$. 
The set of all assets is given by $\mathcal{A} = \bigcup_{c \in \mathcal{C}} \mathcal{A}(c)$. 
We denote by $\mathcal{A}_t(c)$ the set of all $c$-type assets listed in the market on $t \in \mathbb{N}$ (if an asset is synthetic, its listing is defined as the listing of all its constituent assets).
We define $\mathcal{L}_t(c) \subset \mathcal{A}_t(c)$ as the subset of $c$-type assets that are traded by other investors on $t\in\mathbb{N}$ and that do not expire before the market closes on $t+1$. 
Finally, we define $\mathcal{U}_t \subset \mathcal{A}$ as the trading universe for $t\in\mathbb{N}$, i.e., the set of all assets eligible to be invested by the agent on $t$. 
The agent herein is any entity employing the proposed method.

\textbf{Pricing of Assets}.  
For an asset $a \in \mathcal{A}$, we denote its  market price at $\tau \in \mathcal{T}$ by $P_{\tau}(a) \in \mathbb{R}$. 
A present value function $\Pi_{\tau}$ is a linear functional mapping a cash flow to a real number that represents its value at $\tau\in\mathcal{T}$ and exists under the AF assumption (see \citet{skiadas2024theoretical}).
In this paper, we consider only cash flows, each comprising a single flow at a single time point, as inputs to $\Pi_{\tau}$.
Hence, for simplicity, we denote the present value function as a function $\Pi_{\tau}(K;M)$ from a cash flow value $K$ at $M\in\mathcal{T}$ to a real number. 
Since $\Pi_{\tau}$ is a linear functional, we have 
\begin{equation}
\Pi_{\tau}(\lambda K+\lambda' K';M)=\lambda \Pi_{\tau}(K;M)+\lambda'\Pi_{\tau}( K';M) \label{eq:pv_linearity}
\end{equation}
for all $\lambda,\lambda'\in\mathbb{R}$.
If the risk-free interest rate is assumed to be a constant, 
we have 
\begin{equation}
\Pi_{\tau}(K;M)=Ke^{-r_f(M-\tau)}, \label{eq:pv_constant}
\end{equation}
where $r_f\in \mathbb{R}$ is a constant risk-free interest rate with continuous compounding \citep{hull2013fundamentals}.






\textbf{Underlying Instrument}. 
An option is defined on an underlying instrument, such as a stock. 
In this paper, we consider a single underlying instrument, whose  price at $\tau\in\mathcal{T}$ is denoted by $S_{\tau}$. 

\textbf{Put and Call Options}. 
A European put (call) option is the right to sell (buy) an underlying instrument $UI$ at a strike price $K$ upon maturity $M\in\mathcal{T}$. 
A put or call option is fully specified by parameters $UI$, $M$, and $K$.
As we focus on a single underlying instrument, we omit $UI$ for simplicity.
Thus, we write put and call options as $(PT; M, K)$ and $(CL; M, K)$, respectively.
An option yields a payoff only at maturity $M$: $\max\{K - S_M,0\}$ for a put and $\max\{S_M - K,0\}$ for a call; it is zero at all other times.



We now formalize notations for maturities and strike prices. 
We define $\mathcal{M} \subset \mathcal{T}$ and $\mathcal{K} \subset \mathbb{R}$ as the sets of all maturities and strike prices, respectively.
For $c \in \{PT, CL, SL\}$, for $a \in \mathcal{A}(c)$, we let $M_a$ and $K_a$ denote its maturity and strike price, respectively.
Then, for $t\in\mathbb{N}$, we define $\mathcal{L}_t(c; M) = \{ a' \in \mathcal{L}_t(c) : M_{a'} = M \}$ as the set of all type-$c$ assets in $\mathcal{L}_t(c)$ with maturity $M$.
We define $\mathcal{M}_t(c) = \{M_{a'} : a' \in \mathcal{L}_t(c) \}$ as the set of maturities of all type-$c$ assets in $\mathcal{L}_t(c)$.

Under the AF assumption, the put-call parity establishes a relationship between put and call option prices with the same maturity and strike price as follows: for any $M\in\mathcal{M}, K \in \mathcal{K}, \tau\in\mathcal{T}$, 
\begin{equation}
P_{\tau}(PT;M,K)+S_{t} =P_{\tau}(CL;M,K)+\Pi_{\tau}(K;M). \label{eq:put_call_parity}
\end{equation}

\textbf{Synthetic-Longs}. 
A synthetic-long $(SL;M,K)\in \mathcal{A}(SL)$ is a synthetic asset consisting of one long contract of $(CL;M,K)$ and one short contract of $(PT;M,K)$. 
Its price at $\tau\in\mathcal{T}$ is 
\begin{equation}
P_{\tau}(SL;M,K)=-P_{\tau}(PT;M,K)+P_{\tau}(CL;M,K). \label{eq:sl_price}
\end{equation}
Thus, we can rewrite the put-call parity \eqref{eq:put_call_parity} as:
\begin{equation}
S_{\tau}=P_{\tau}(SL;M,K)+\Pi_{\tau}(K;M). \label{eq:put_call_parity2}
\end{equation}
Furthermore, $(SL;M,K)$ yields a payoff of
\begin{equation}
\text{Payoff}_{\tau}({SL;M,K})
=\begin{cases}
    S_\tau - K &\text{if}~\tau=M \\
    0 &\text{if}~\tau \neq M
\end{cases}. \label{eq:sl_payoff}
\end{equation}
In this paper, we only invest in synthetic-longs and their combinations ($LS, SA$). 

\subsection{Notation and Assumptions}
\label{subsec:setup_notation}
\textbf{Notations.}
For finite $X \subset \mathbb{R}$, we define $\min^k X$ as the $k$-th smallest element in $X$. 
For simplicity, we write $\min X := \min^1 X$. 
The notations $\max^k X$ and $\max X := \max^1 X$ are defined analogously.
Given a function $f: X \to \mathbb{R}$, we define $\arg\min_x^k \{ f(x) : x \in X \}$ as the set of elements in $X$ for which $f(x)$ equals the $k$-th smallest value of $f(X)$. 
We similarly define $\arg\min_x := \arg\min_x^1$, as well as $\arg\max_x^k$ and $\arg\max_x$ analogously. 
Lastly, for $\mathcal{A}' \subset \mathcal{A}$ and vectors $\textbf{x}_a\in\mathbb{R}^p$, $a \in \mathcal{A}'$, 
$[\textbf{x}_a^T]_{a \in \mathcal{A}'}\in\mathbb{R}^{\vert \mathcal{A}' \vert \times p}$ denotes a row-stacked matrix ordered lexicographically. 

We utilize two datasets, $\mathcal{D}^{tr} = \{D_t^{tr} : t \in \mathbb{N}\}$ and $\mathcal{D}^{ar} = \{D_t^{ar} : t \in \mathbb{N}\}$ with
\begin{align}
 D_t^{tr} &= \left([\textbf{x}_{tr,a,t-1}^T]_{a\in\mathcal{L}_{t-1}},[\tilde{I}_{tr,a,t}]_{a\in\mathcal{L}_{t-1}} \right),\\
 D_t^{ar} &= \left([\textbf{x}_{ar,a,t-1}^T]_{a\in\mathcal{U}_{t}}, G_{t-1}(p_{dg}), [y_{a,o(t)}]_{a\in\mathcal{U}_{t}}\right).
\end{align}
Each $D_t^{tr}$ and $D_t^{ar}$ is constructed to predict tradability (\ssref{subsec:universe}) and arbitrage (\ssref{subsec:prediction}) on $t$, respectively, using information up to the date of $t-1$. 
Components $\textbf{x}_{tr,a,t-1} \in \mathbb{R}^{p_{tr}}, \textbf{x}_{ar,a,t-1} \in \mathbb{R}^{p_{ar}}$, $G_{t-1}(p_{dg})$, $\tilde{I}_{tr,a,t} \in \{0,1\}, y_{a,o(t)}$ will be specified in 
Sections \ref{sec:method} and \ref{sec:experiments}. 

For the prediction models, we repeatedly split the datasets $\mathcal{D}^{tr}$ and $\mathcal{D}^{ar}$ into train, validation, and test subsets. 
Let $\mathcal{T}_{fit}=\{t_{1},...,t_{n_{fit}},t_{n_{fit}+1}\}\subset \mathbb{N}\cup\{\infty\}$ with  $t_{i}<t_{i+1}~\text{and}~t_{n_{fit}+1}=\infty$ denote the starting dates of test splits, where $n_{fit} \in \mathbb{N}$. 
Each index $i \in [1:n_{fit}]$ defines a train-test round. 
Next, for each $k\in\{tr,ar\}$, we define the split time index sets $\mathcal{T}_{test,i}^k, \mathcal{T}_{train,i}^k,\mathcal{T}_{val,i}^k$ as  
\begin{align}
    &\mathcal{T}_{test,i}^k = [t_{i}:t_{i+1}), \label{eq:time_set_test}\\
    &\mathcal{T}_{train,i}^k \cup \mathcal{T}_{val,i}^k = [1:t_{i}),\\
    &\mathcal{T}_{train,i}^k \cap \mathcal{T}_{val,i}^k =\emptyset,
\end{align}
where $\mathcal{T}_{val,i}^k$ is randomly sampled from $[1:t_i)$ using a split ratio $p_{val}=|\mathcal{T}_{val,i}^k|/(|\mathcal{T}_{val,i}^k|+|\mathcal{T}_{train,i}^k|)$. 
Then, we define data splits as $\mathcal{D}_{train,i}^k = \{D_t^k:t \in \mathcal{T}_{train,i}^k\}$,  $\mathcal{D}_{val,i}^k = \{D_t^k:t \in \mathcal{T}_{val,i}^k\}$, and $\mathcal{D}_{test,i}^k = \{D_t^k:t \in \mathcal{T}_{test,i}^k\}$. 
Lastly, we group them as $\mathcal{D}_i^k= (\mathcal{D}_{train,i}^k,\mathcal{D}_{val,i}^k,\mathcal{D}_{test,i}^k)$. 
Given $\mathcal{D}_i^k$, 
we train prediction models on $\mathcal{D}_{train,i}^k$, select the best model based on $\mathcal{D}_{val,i}^k$, and evaluate it on $\mathcal{D}_{test,i}^k$\footnote{
These splits support proper model evaluation. 
Each $\mathcal{D}_i^k$ contains disjoint train, validation, and test splits. 
In addition, \eqref{eq:time_set_test} ensures test datasets are non-overlapping and jointly cover $[t_{1}:\infty)$: $\forall i \neq j,~ \mathcal{T}_{test,i}^k \cap \mathcal{T}_{test,j}^k = \emptyset,~\text{and}~\bigcup_{i=1}^{n_{fit}}\mathcal{T}_{test,i}^k = [t_{1}: \infty)$. 
Moreover, since $t<t'$, $\forall t\in \mathcal{T}_{train,i}^k \cup \mathcal{T}_{val,i}^k, \forall t'\in \mathcal{T}_{test,i}^k$, the splits avoid look-ahead bias. 
}.

\textbf{Assumptions}. 
Throughout the paper, we assume the following:
\begin{itemize}
    \item Positions may consist of any number of option contracts. 
    \item The agent's trades have no impact on markets. 
    \item Option margin requirements can be ignored due to sufficiently large long-term holdings of substitute assets, such as stocks, bonds, and funds. 
    \item Short selling is permitted without limit.  
    \item Unlimited borrowing and lending are allowed at the risk-free rate. 
    \item There are no taxes or transaction costs. 
    \item There are no restrictions on market orders during the opening market.
\end{itemize}
\section{Proposed Method}
\label{sec:method}

\begin{figure*}[t] 
    \centering
    \includegraphics[width=\textwidth]{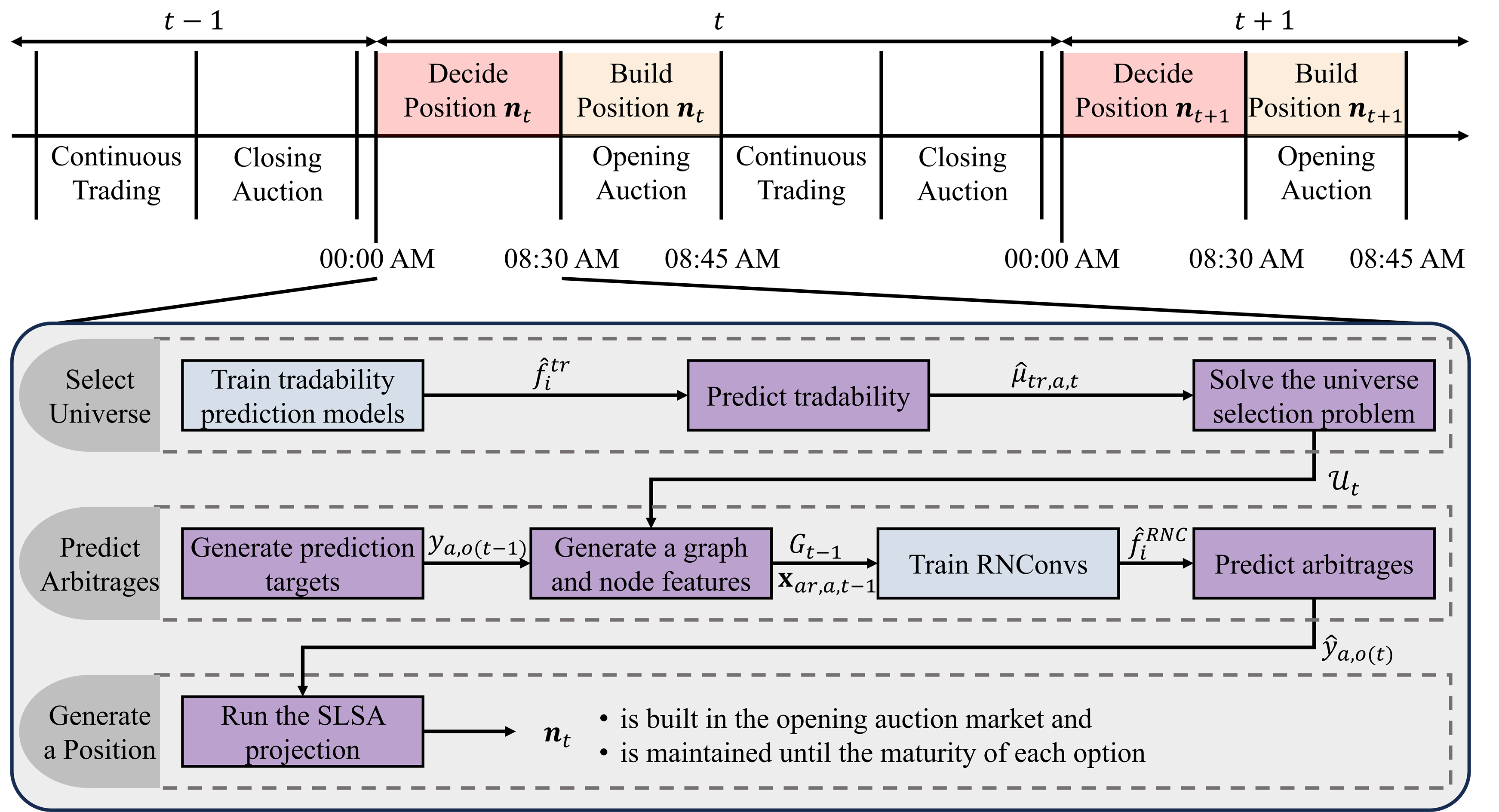}
    \caption{Timeline. \textnormal{The blue boxes are executed only on the dates that belong to $\mathcal{T}_{fit}$ and are skipped on all other days. 
    On each  $t\in [t_{i}:t_{i+1})$, the most recently trained models, $\hat{f}^{tr}_{i}$ and $\hat{f}_{i}^{RNC}$, are used.
    }}
    \label{fig:timeline}
\end{figure*}

The final goal of our method is to determine the position $\textbf{n}_t=\left[n_{a,t}\right]_{a\in\mathcal{U}_t}\in\mathbb{R}^{|\mathcal{U}_t|}$ for each $t\in\mathbb{N}$, 
where each $n_{a,t}$ represents the number of contracts for asset $a$. 
For $t\in\mathbb{N}$, our method \proman{1}~ determines the trading universe $\mathcal{U}_t$\footnote{While the major contributions of our two-step approach are in \proman{2} and \proman{3}, we begin with \proman{1} to provide a foundation for their appropriate functioning.}, \proman{2}~ predicts arbitrage opportunities $\hat{y}_{a,o(t)}$, and \proman{3}~ derives a low-risk position $\textbf{n}_t$. 
The three tasks are conducted for every $t$, as depicted in \fref{fig:timeline} based on times from the Korean market. 
The following sections elaborate on our approach for each task.

The determined position $\textbf{n}_t$ is executed in the opening auction and is held until the expiration of each option included in $\textbf{n}_t$. 
These processes are conducted on every trading date.

Note that the three tasks of the decision process for $\mathbf{n}_t$ rely solely on information available up to the date of $t-1$, excluding any data from $t$ or beyond.
This is because using information available after the decision process on $t$ to determine $\mathbf{n}_t$ incurs look-ahead bias \footnote{See \citet{alonso2024look} and \citet{ter2001eliminating} for a detailed discussion of look-ahead bias.}.  

\subsection{Trading Universe Selection Model}
\label{subsec:universe}

A universe selection method is necessary to ensure back-testing validity, practical feasibility, and proper graph construction. 
For example, deep ITM/OTM options are rarely traded and are often missing from historical data, leading to unavailable prices~\citep{anand2007stealth, blasco2010does}.
Thus, synthetic-longs involving such options cannot be evaluated. 


To address such necessity, for $t\in\mathbb{N}$, we solve \eqref{univ:obj}--\eqref{univ:end} to construct $\mathcal{U}_t\subset \mathcal{L}_{t-1}(SL)$, and on $t$, we only trade the synthetic-longs in $\mathcal{U}_t$.   
Since the trading status on $t$---used to define $\mathcal{L}_{t}(SL)$---is unknown when $\mathcal{U}_t$ is determined, $\mathcal{U}_t$ is defined as a subset of $\mathcal{L}_{t-1}(SL)$, not $\mathcal{L}_{t}(SL)$. 
Moreover, arbitrages on $t$ are predicted only for synthetic-longs traded on $t-1$ (\ssref{subsec:prediction}). 

For each $a \in \mathcal{L}_{t-1}(SL)$, $u_{a,t}\in\{0,1\}$ indicates whether asset $a$ is included in $\mathcal{U}_t$. 
The optimal solution $u_{a,t}^*$ yields universe $\mathcal{U}_t=\{a\in\mathcal{L}_{t-1}(SL):u_{a,t}^*=1\}$. 
Variable $v_{M,t}$ indicates whether synthetic-longs with maturity $M$ are included in $\mathcal{U}_t$.

\begin{align}
&\text{max}  \sum_{a \in \mathcal{L}_{t-1}(SL)} \hat{\mu}_{tr,a,t} u_{a,t}   \label{univ:obj} \\ 
&\text{s.t.}  \nonumber \\
& \sum_{a\in\mathcal{L}_{t-1}(SL) } u_{a,t} = p_{univ},  \label{univ:size}\\
& u_{N_{atm,t-1}(a),t} \geq u_{a,t}, \notag\\
&\quad\quad\quad\quad \forall a \in \mathcal{L}_{t-1}(SL) \setminus \ ATM_{t-1}(SL), \label{univ:near_atm}\\
& \edit{2 v_{M,t} \leq \sum_{a\in\mathcal{L}_{t-1}(SL;M)} u_{a,t} \leq |\mathcal{L}_{t-1}(SL;M)|\cdot v_{M,t},} \notag\\
&\quad\quad\quad\quad\forall M \in \mathcal{M}_{t-1}(SL),  \label{univ:longshort1}\\
& u_{a,t}+u_{a',t}\leq 1, ~ \forall (a,a')\in N_{far,t-1}, \label{univ:far}\\
& u_{a,t}\in\{0,1\}, ~\forall a \in \mathcal{L}_{t-1}(SL), \\
&v_{M,t} \in \{0,1\}, ~ \forall M \in \mathcal{M}_{t-1}(SL). \label{univ:end}
\end{align}

Objective function \pref{univ:obj} aims to select $\mathcal{U}_t$ so as to maximize the number of synthetic longs traded by other investors on $t$, i.e., $\text{max}_{\mathcal{U}_{t}\subset\mathcal{L}_{t-1}(SL)} |\mathcal{U}_{t}\cap\mathcal{L}_{t}(SL)|$. 
Since trading status on $t$ is unknown when determining $\textbf{n}_t$ (\fref{fig:timeline}), we define $I_{tr,a,t} \in \{0,1\}$ as a Bernoulli random variable with $\mu_{tr,a,t}$ indicating whether asset $a$ is traded on $t$ and maximize $\mathbb{E}\sum_{a \in \mathcal{L}_{t-1}(SL)} I_{tr,a,t} u_{a,t}$ instead of the exact number.   
We assume there exists a function $f^{tr}_{i}(\textbf{x}_{tr,a,t-1})$ 
such that $t\in [t_{i}:t_{i+1})$, $\mu_{tr,a,t}=f^{tr}_{i}(\textbf{x}_{tr,a,t-1})+\epsilon_{tr,a,t}$, $\mathbb{E}\epsilon_{tr,a,t}=0$, and $\epsilon_{tr,a,t}$ and $\epsilon_{tr,a',t}$ are independent for all $a \neq a'$.
Here, $\textbf{x}_{tr,a,t-1}$ is a deterministic feature vector.
Then, $\mathbb{E}I_{tr,a,t}=\mathbb{E}\mathbb{E}[I_{tr,a,t}|\mu_{tr,a,t}]=\mathbb{E}\mu_{tr,a,t}=f^{tr}_{i}(\textbf{x}_{tr,a,t-1})$.
We estimate $f^{tr}_{i}$ and $\mu_{tr,a,t}$ as a machine-learning model $\hat{f}^{tr}_{i}$ and $\hat{\mu}_{tr,a,t}=\hat{f}^{tr}_{i}(\textbf{x}_{tr,a,t-1})$, respectively, described later. 
Finally, we have $\mathbb{E}\sum_{a \in \mathcal{L}_{t-1}(SL)} I_{tr,a,t} u_{a,t}=\sum_{a \in \mathcal{L}_{t-1}(SL)} \mu_{tr,a,t} u_{a,t} \approx \sum_{a \in \mathcal{L}_{t-1}(SL)} \hat{\mu}_{tr,a,t} u_{a,t}$.

We train $\hat{f}^{tr}_{i}$ and predict $\hat{\mu}_{tr,a,t}$ using $\mathcal{D}_i^{tr}$.  
Multiple models are trained on $\mathcal{D}_{train,i}^{tr}$ using the target $\tilde{I}_{tr,a,t}$—a realization of $I_{tr,a,t}$. 
The best model $\hat{f}^{tr}_{i}$ is selected using $\mathcal{D}_{val,i}^{tr}$, and it generates predictions $\hat{\mu}_{tr,a,t} = \hat{f}^{tr}_{i}(\textbf{x}_{tr,a,t-1})$ on $\mathcal{D}_{test,i}^{tr}$.

We next discuss the constraints. Constraint \pref{univ:size} ensures the cardinality of $\mathcal{U}_t$ becomes $p_{univ}$, which is a user-specified parameter. 
Constraint \pref{univ:near_atm} ensures near-ATM assets are selected.  
Let $ATM_t(SL)=\{a : a \in \text{argmin}_{(SL;M^*,K)\in\mathcal{L}_t} |S_t-K|,~\forall M^*\in\mathcal{M}_t(SL)\}$ denote the set of all the nearest ATM assets for each maturity. 
For $a \in \mathcal{L}_t(SL)$, let $ N_{atm,t}(a)=\text{argmin}_{a'\in \mathcal{L}_{t}(SL)} \{|K_a-K_{a'}|:~M_a=M_{a'}~\text{and}~|S_t-K_{a'}|<|S_t-K_a|\}$ denote its same-maturity nearest neighbor closer to ATM.  
Then, an asset $a$ is excluded from $\mathcal{U}_t$ if the universe excludes its same-maturity nearest neighbor closer to ATM. 
Since 
$ATM_{t-1}(SL)$ admit no neighbors closer to ATM, they are excluded in \pref{univ:near_atm}. 
Because we expect near-ATM options at $t-1$ are likely to be also traded on $t$, \pref{univ:near_atm} is included.


Constraint \pref{univ:longshort1} enforces that if any maturity-$M$ option is selected, then at least one of the other maturity-$M$ options must be selected. 
Since the positions (cf. \ssref{subsec:slsa}) require both long and short positions on same-maturity options, isolated inclusion of a maturity-$M$ asset renders it unusable. 
Therefore, the right inequality of \pref{univ:longshort1} ensures that if a maturity-$M$ asset is selected (i.e., $u_{a,t}=1$ for some $a\in\mathcal{L}_{t-1}(SL;M)$), then the maturity $M$ itself is selected (i.e., $v_{M,t}=1$).
Moreover, another maturity-$M$ asset must be selected (i.e., $u_{a',t}=1$ for some $a'\in\mathcal{L}_{t-1}(SL;M)$ with $a\neq a'$) by the left inequality of \pref{univ:longshort1}. 


Constraint \pref{univ:far} chooses one of the two same-maturity neighborhood assets that are far away from each other. 
Let $N_{far,t}=\{(a,a')\in\mathcal{L}_{t}(SL)\times\mathcal{L}_{t}(SL): M_a=M_{a'}, |K_a-K_{a'}|=\max\{\min_{a''\in\mathcal{L}_{t}(SL;M_a)}|K_{a''}-K^*|:K^*\in\{K_a,K_{a'}\}\}>\Delta K_{max}\}$ denote the set of the pairs of the same-maturity assets of which strike price differences are large.
Parameter $\Delta K_{max}$ quantifies the notion of being ``far." 
Since a node in our prediction graphs (\ssref{subsec:prediction}) is connected to a fixed number of nearest same-maturity neighbors, overly distant connections may arise without this constraint. 
To avoid such cases, \pref{univ:far} is imposed. 



\subsection{Arbitrage Prediction}
\label{subsec:prediction}



\subsubsection{Arbitrage Prediction Problem Statement}
\label{subsubsec:arb_problem}
The arbitrage prediction problem addressed in this section is a node-level regression problem.  
Given a graph and its node features $\textbf{x}_{ar,a,t-1}$, our goal is to predict the novel arbitrage opportunity $y_{a,o(t)}$ corresponding to each node $a\in\mathcal{U}_t$ using the new proposed architecture, RNConv.
We next discuss these components. 


\textbf{Prediction Target.}
Our arbitrage prediction model aims to predict $y_{a,o(t)}$ for $t\in\mathbb{N}$, $a\in\mathcal{U}_t$. 
For all $\tau\in\mathcal{T}$, $a=(SL;M_a,K_a)\in\mathcal{U}_{\lfloor \tau \rfloor}$\footnote{Here, the subscript $\lfloor \tau \rfloor \in \mathbb{N}$ in $\mathcal{U}_{\lfloor \tau \rfloor}$ implies the trading date corresponding to the time point $\tau$, according to Section~\ref{subsec:setup_env}.}, $y_{a,\tau}$ is defined as
\begin{equation}
    y_{a,\tau}=\delta_{a,\tau}-\bar{\delta}_{M_a,\tau} \label{eq:pred_target} \\ 
\end{equation}
where
\begin{equation}    
    \delta_{a,\tau}=\frac{ S_{\tau}- P_{\tau}(SL;M_a,K_a)}{K_a}, \label{eq:unit_bond} \\
\end{equation}
\begin{equation} 
    \bar{\delta}_{M_a,\tau}=\frac{\sum_{a'\in\mathcal{U}_{\lfloor \tau \rfloor}\cap\mathcal{L}_{\lfloor \tau \rfloor-1}(SL;M_a)} \delta_{a',\tau}}{\vert\mathcal{U}_{\lfloor \tau \rfloor}\cap\mathcal{L}_{\lfloor \tau \rfloor-1}(SL;M_a)\vert}. \label{eq:dbar}
\end{equation}
The price $P_{\tau}(SL;M_a,K_a)$ is computed from the put and call prices via \eqref{eq:sl_price}, and $\bar{\delta}_{M_a,\tau}$ is the average of $\delta_{a',\tau}$ over the maturity-$M_a$ assets in the universe. 
We designate the market open time $o(t)$ as the target time point for prediction (i.e., we predict $y_{a,o(t)}$ amongst many $y_{a,\tau}$).

Variable $\delta_{a,\tau}$ represents the price at $\tau$ of a zero-coupon bond maturing at $M_a$ with face value 1.  
By \eqref{eq:unit_bond}, $\delta_{a,\tau}$ is the cost of building a position with $1/K_a$ units of the underlying instrument and $-1/K_a$ synthetic-long contracts.
When the positions of $\delta_{a,\tau}$ are cleared at $M_a$, the resulting payoff is 
\begin{equation}
    \text{Payoff}_{M_a}(\delta_{a,\tau})= \frac{S_{M_a}}{K_a}-\frac{S_{M_a}-K_a}{K_a}=1, \label{eq:d_payoff}
\end{equation}
which follows from \eqref{eq:sl_payoff}. 
Hence, $ \delta_{a,\tau}$ serves as the \textbf{d}iscounted price—or equivalently, the \textbf{d}iscount factor—of a zero-coupon bond with face value 1 and maturity $M_a$.

Arbitrage opportunities can be found in $\delta_{a,\tau}$'s. 
Under the AF assumption, $\delta_{a,\tau}=\delta_{a',\tau}$ should hold for any $a,a'\in\mathcal{L}_{\lfloor \tau \rfloor-1}(SL;M)$, as both represent the same zero-coupon bond\footnote{Under the AF assumption, the law of one price holds \citep{skiadas2024theoretical}.}. 
However, in reality, we have $\delta_{a,\tau} \neq \delta_{a',\tau}$ in many cases, as shown in \fref{fig:arb_time_series} since the blue and red curves should be zero if $\delta_{a,\tau}=\delta_{a',\tau}$ were true. 
Hence, arbitrage opportunities can be found in $\delta_{a,t}$'s.

\begin{figure*}[t] 
    \centering
    \includegraphics[width=\textwidth]{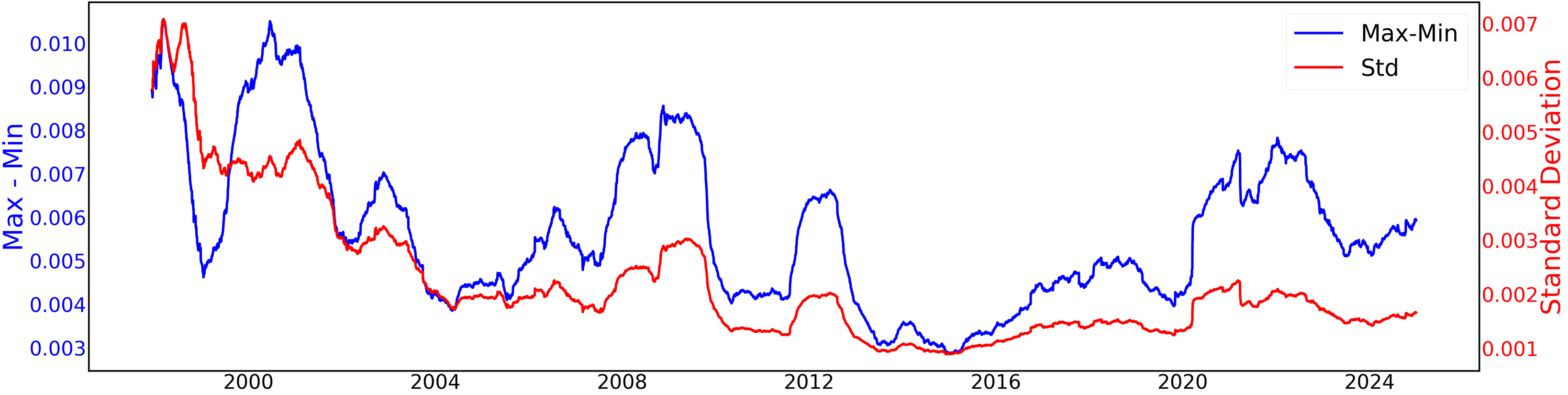}
    \caption{Time-Series of Rolling 252-Trading-Day Avg. of Max-Min and Std. of $\delta_{a,t,o}$ over Nearest-Term, Near-ATM Options. 
    \textnormal{
    On each $t$, the range and standard deviation of $\delta_{a,t,o}$ are computed from the open price data of the KOSPI 200 index and near-ATM, nearest-maturity options ($S_{t,o}/K \in (0.95, 1.05)$). Then, the 252-day rolling averages are plotted.}}
    \label{fig:arb_time_series}
\end{figure*}

To isolate arbitrages embedded in $\delta_{a,\tau}$, we subtract the average $\bar{\delta}_{M_a,\tau}$ from $\delta_{a,\tau}$ and define $y_{a,\tau}$ as in \pref{eq:pred_target}. 
Predicting $\delta_{a,\tau}$ itself would conflate the bond price with arbitrages, as implied by \pref{eq:d_payoff}.  
Thus, we remove the bond by subtracting $\bar{\delta}_{M_a,\tau}$, which serves as an estimate of the bond price.

If $y_{a,\tau} > 0$ were known in advance, arbitrages could be exploited by shorting $\delta_{a,\tau}$ and longing $\bar{\delta}_{M_a,\tau}$ at $\tau$, yielding an arbitrage profit of $|y_{a,\tau}|$ at $\tau$. 
The reverse would symmetrically yield an arbitrage profit.  
Hence, predicting \pref{eq:pred_target} directly identifies pure arbitrages.

Proposition~\ref{prop:pred_target} shows that $y_{a,\tau}$ comprises pure arbitrages, and that $\bar{\delta}_{M_a,\tau}$ estimates the expected discount factor. 
This follows from the fact that, under the AF assumption,  $y_{a,\tau}$ vanishes, and $\bar{\delta}_{M_a,\tau}$ becomes the present value of a unit zero-coupon bond maturing at $M_a$ in Proposition~\ref{prop:pred_target}. 
\begin{proposition} \label{prop:pred_target}
     If the market is AF, then the following holds:
        \begin{align}
            &y_{a,\tau}=0 \label{eq:arb_free_y},\\
            &\delta_{a,\tau}=\bar{\delta}_{M_a,\tau}=\Pi_{\tau}(1;M_a). \label{eq:arb_free_d}
        \end{align}
\end{proposition}

\begin{proof}{Proof}
From \eqref{eq:put_call_parity2} and \eqref{eq:pv_linearity}, we obtain $S_{\tau} - P_{\tau}(SL; M_a, K_a) = K_a \cdot \Pi_{\tau}(1, M_a)$.  
Substituting this result into \eqref{eq:unit_bond} yields $\delta_{a,\tau} = \Pi_{\tau}(1; M_a)$ for all $a \in \mathcal{U}_{\lfloor \tau \rfloor} \cap \mathcal{L}_{\lfloor \tau \rfloor -1}(SL; M_a) $, establishing \eqref{eq:arb_free_d}.  
Thus, \eqref{eq:arb_free_y} follows by substituting \eqref{eq:arb_free_d} into \eqref{eq:pred_target}.
\Halmos
\end{proof}


\textbf{Graph Construction.}
To predict $y_{a,o(t)}$ for $a\in\mathcal{U}_{t}$ and $t\in\mathbb{N}$, a graph $G_{t-1}(p_{dg})=(\mathcal{U}_{t},L_{t-1}(p_{dg}))$ is considered  where  
$\mathcal{U}_{t}$ is the node set\footnote{
The subscript $t$ of $\mathcal{U}_{t}$ in $G_{t-1}(p_{dg})$ does not imply the use of future information at $t$ to predict $y_{a,o(t)}$, as clarified in \ssref{subsec:universe}; it merely indicates that the assets in $\mathcal{U}_{t}$ are traded on $t$.
}, $L_{t-1}(p_{dg})\subset\mathcal{U}_{t}\times\mathcal{U}_{t}$ is the edge set, and $p_{dg}\in[0,1]$ is a hyperparameter. 
If $(a,a')\in L_{t-1}(p_{dg})$, one of the following is satisfied:
\begin{enumerate}
    \item $M_a=M_{a'}$ and $K_{a} \in argmin^{k'}_{a''} \{|K_{a'}-K_{a''}|:a''\in\mathcal{U}_t,M_{a''}=M_{a'}\}$ where $k'$ is the nearest multiple of two to $p_{dg} \cdot \left| \{ a'' \in \mathcal{U}_t : M_{a''} = M_{a'} \} \right|$ (set to one if the result is zero).
    \item $K_a=K_{a'}$ and $M_{a} \in argmin^{k'}_{a''} \{|M_{a'}-M_{a''}|:a''\in\mathcal{U}_t,K_{a''}=K_{a'}\}$ where $k'$ is the nearest multiple of two to $p_{dg} \cdot \left| \{ a'' \in \mathcal{U}_t : K_{a''} = K_{a'} \} \right|$ (set to one if the result is zero).
\end{enumerate}
The first condition requires that $a$ and $a'$ share the same maturity, and that $a$ is among the nearest-$k'$ neighbors of $a'$ by strike price.
The second applies the same logic with strike price and maturity reversed.
Since strike prices increase by a fixed increment (e.g., 2.5 for KOSPI 200 options), each option typically has two nearest same-maturity neighbors. 
Thus, $k'$ is set as a multiple of two and applied analogously for maturity.
\textbf{Node Features.}
For $t\in\mathbb{N}$, we use the following as node features $\textbf{x}_{ar,a,t-1}$ to predict arbitrage $y_{a,o(t)}$ where $c(t),h(t),l(t)\in\mathcal{T}$ map $t\in\mathbb{N}$ into the values in $\mathcal{T}$ corresponding to the time points when the market closes, achieves high, and becomes low, respectively: 
\begin{itemize}
    \item Moneyness ($S_{c(t-1)} - K_a$) and days remaining until $M_a$, 
    \item $y_{a,c(t-1)},y_{a,o(t-1)}$: the closing and opening values of $y$ on $t-1$, 
    \item $\hat{y}_{a,h(t-1)}, \hat{y}_{a,l(t-1)}$: the estimated high and low values of $y$ on $t-1$. The high (low) estimate uses the high (low) put and low (high) call prices, 
    \item $y_{a,c(t-1)}-y_{a,c(t-2)}$: the change in the closing target value between $t-1$ and $t-2$, 
    \item $\hat{\sigma}_{im,a,t-1}^{PT},\hat{\sigma}_{im,a,t-1}^{CL}$: the implied volatilities of the put and call options of $a$ on $t-1$.
\end{itemize}


\subsubsection{Revised Neural Oblivious Decision Ensemble Graph Convolution (RNConv)}
\label{subsubsec:rnconv}
We first present the Revised NODE (RNODE) layer, which underpins RNConv. Then, we propose the RNConv layer and the overall RNConv architecture.

\textbf{RNODE Layer.} 
To enhance NODE performance, we present the RNODE layer by replacing $b_i,\kappa_i$, and $entmax_{\alpha}(\hat{\textbf{s}}_i)$ in \pref{eq:dodt3} with batch normalization \citep{ioffe2015batch}, a low-rank cross network \citep{wang2021dcn}, and a multi-layer perceptron.
Let $d$ and $n_{rdt}$ be hyperparameters denoting the depth and number of trees in an RNODE layer, respectively, and $\gamma$ be the dichotomizing hyperparameter.
Given an input $\textbf{x} \in \mathbb{R}^p$, the output of \textbf{RN}ODE \textbf{L}ayer is defined as
\begin{equation}
    g^{RNL} (\textbf{x}) = \left[h_{rdt,1}(\textbf{x})~...~h_{rdt,n_{rdt}}(\textbf{x}) \right] \in \mathbb{R}^{1\times n_{rdt}} \label{eq:rdodt1} 
\end{equation}
where for~$j\in\{1,2,...,n_{rdt}\}, k\in\{1,2,..,{d}\}$,
\begin{align}
    &h_{rdt,j}(\textbf{x})=\sum_{i_1,...,i_d\in\{1,2\}^d} R_{j,i_1,...,i_d} \cdot C_{j,i_1,...,i_d} (\textbf{x}) \in \mathbb{R}, \label{eq:rdodt2}\\
    &C_j(\textbf{x})=\begin{bmatrix} c_{j,1}(\textbf{x}) \\ 1 - c_{j,1}(\textbf{x}) \end{bmatrix} \otimes ... \otimes \begin{bmatrix} c_{j,d}(\textbf{x}) \\ 1 - c_{j,d}(\textbf{x}) \end{bmatrix}, \label{eq:rdodt3}\\
    &c_{j,k}(\textbf{x})=\sigma_{\alpha}(\gamma \cdot VBN(x'_{j,k})) \in \mathbb{R}, \label{eq:rdodt4}\\
    &\begin{bmatrix} \textbf{x}'_1 \\ \textbf{x}'_2 \\ ... \\ \textbf{x}'_{n_{rdt}} \end{bmatrix} 
    =MLP(CrossNet(\textbf{x}))\in \mathbb{R}^{n_{rdt}d},  \label{eq:rdodt5}\\
    &\textbf{x}'_j=[x'_{j,1}~x'_{j,2}~...~x'_{j,d}]^T \in \mathbb{R}^{d \times 1}. \label{eq:rdodt6}  
\end{align}
We define $h_{rdt,j}(\textbf{x})$ as a \textbf{r}evised \textbf{d}ifferentiable oblivious \textbf{t}ree. 
Equations \eqref{eq:rdodt2} and \eqref{eq:rdodt3} are derived from \eqref{eq:dodt1} and \eqref{eq:dodt2}, but differ by including a tree index $j$, which accounts for multiple distinct RDTs. 
The function $\sigma_\alpha(\cdot)$ is defined in \pref{eq:dodt4}. 
Figure \ref{fig:rdt} illustrates an RNL layer. 
\begin{figure*}[t] 
    \centering
    \includegraphics[width=\textwidth]{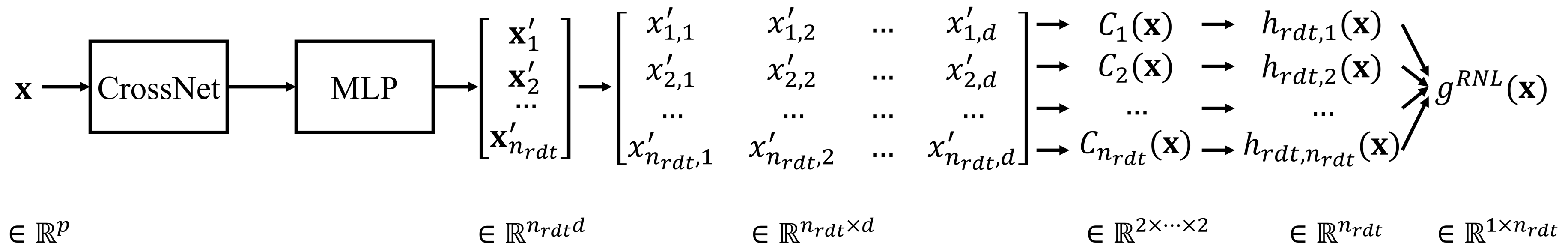}
    \caption{Revised Neural Oblivious Ensemble (RNODE) Layer}
    \label{fig:rdt}
\end{figure*}

RNODE differs from NODE by replacing \eqref{eq:dodt3} with \eqref{eq:rdodt4}–\eqref{eq:rdodt6}.  
In \eqref{eq:rdodt4}, $VBN(\cdot)$ is vanilla batch normalization without learnable parameters, 
while $BN(\cdot)$ includes learnable parameters 
\citep{ioffe2015batch}.
Extending the low-rank cross network \eqref{eq:original_dcn} of \citet{wang2021dcn}, we utilize a batch-normalized variant $CrossNet(\cdot)$ in \eqref{eq:rdodt5}  to enhance optimization stability, defined as
\begin{equation}
    CrossNet(\textbf{x}) = \textbf{z}_{l_{CN}} \in \mathbb{R}^p \label{eq:crossnet_rnode}
\end{equation}
where for $l\in\{0,1,...,l_{CN}-1\}$,
\begin{align}
    & \textbf{z}_{l}= \begin{cases}
        BN(\textbf{x}) \in \mathbb{R}^p ~&\text{if}~l=0\\
        BN(\textbf{x}_{l}) \in \mathbb{R}^p ~&\text{if}~l \geq 1\\
    \end{cases}, \label{eq:crossnet1}\\
    &\textbf{x}_{l+1}=\textbf{x} \odot \left(U_l^{CN}\phi(C_l^{CN}\phi({V_l^{CN}}^T \textbf{z}_l))+\textbf{b}_l^{CN}\right) +\textbf{z}_l.
    \label{eq:crossnet2}
\end{align}
Function $MLP(\cdot)$ is a multilayer perceptron of which last layer does not have a bias and activation function.  
Since the output of $MLP(\cdot)$ passes through $VBN(\cdot)$, the last layer's bias becomes meaningless, so it is removed. 

To encourage balanced branching, we use $VBN(\cdot)$ in \pref{eq:rdodt4}. 
In NODE \citep{popov2019neural}, $b_i$ and $\kappa_i$ determines the branching of a DODT. 
However, depending on the training process, $|b_i|$ can grow arbitrarily large, causing imbalanced branching. 
Likewise, if $|\kappa_i|$ becomes large, $c_i(\textbf{x})$ in \eqref{eq:dodt3} becomes nearly identical across all inputs, rendering branching meaningless.
To address these issues, RNODE replaces $b_i$ and $\kappa_i$ with the batch mean and standard deviation within $VBN(\cdot)$.  
As the batch mean is likely to lie near the center of a batch, branching becomes balanced. 
Normalizing by the batch standard deviation ensures unit variance in $VBN(x'_{l,i})$, which is then scaled by $\gamma \in \mathbb{R}$ to push $c_{j,i}$ toward 0 or 1, addressing the second issue\footnote{
Note that too large $\gamma$ may cause gradient saturation during training. 
}.

To overcome the decision tree's limitation of having an axis-aligned decision boundary, in \pref{eq:rdodt5}, we replace $entmax_{\alpha}(\hat{\textbf{s}}_i)$ with $MLP(CrossNet(\cdot))$. 
In classical decision trees, where each split compares a single variable to a threshold, approximating the true boundary $x_1 - x_2 = 0$ requires many stair-shaped splits (orange in \fref{fig:decision_boundary}). 
In contrast, allowing the direct use of $x_1 - x_2 > 0$ enables a single, more efficient split.
However, NODE uses $\textbf{x}^T entmax_{\alpha}(\hat{\textbf{s}}_{i})-b_i$ in \eqref{eq:dodt3} to mimic the single-variable criterion of classical decision trees.
To address this inefficiency, we replace $\textbf{x}^T entmax_{\alpha}(\hat{\textbf{s}}_{i})-b_i$ with $\textbf{x}^T \hat{\textbf{s}}_{i}-b_i$, which can learn decision boundaries like $x_1 - x_2 = 0$.  
We further extend this linear form to a multilayer perceptron and cross network, enabling the model to learn nonlinear decision boundaries.


\begin{figure*}[t] 
    \centering
    \includegraphics[width=0.4\textwidth]{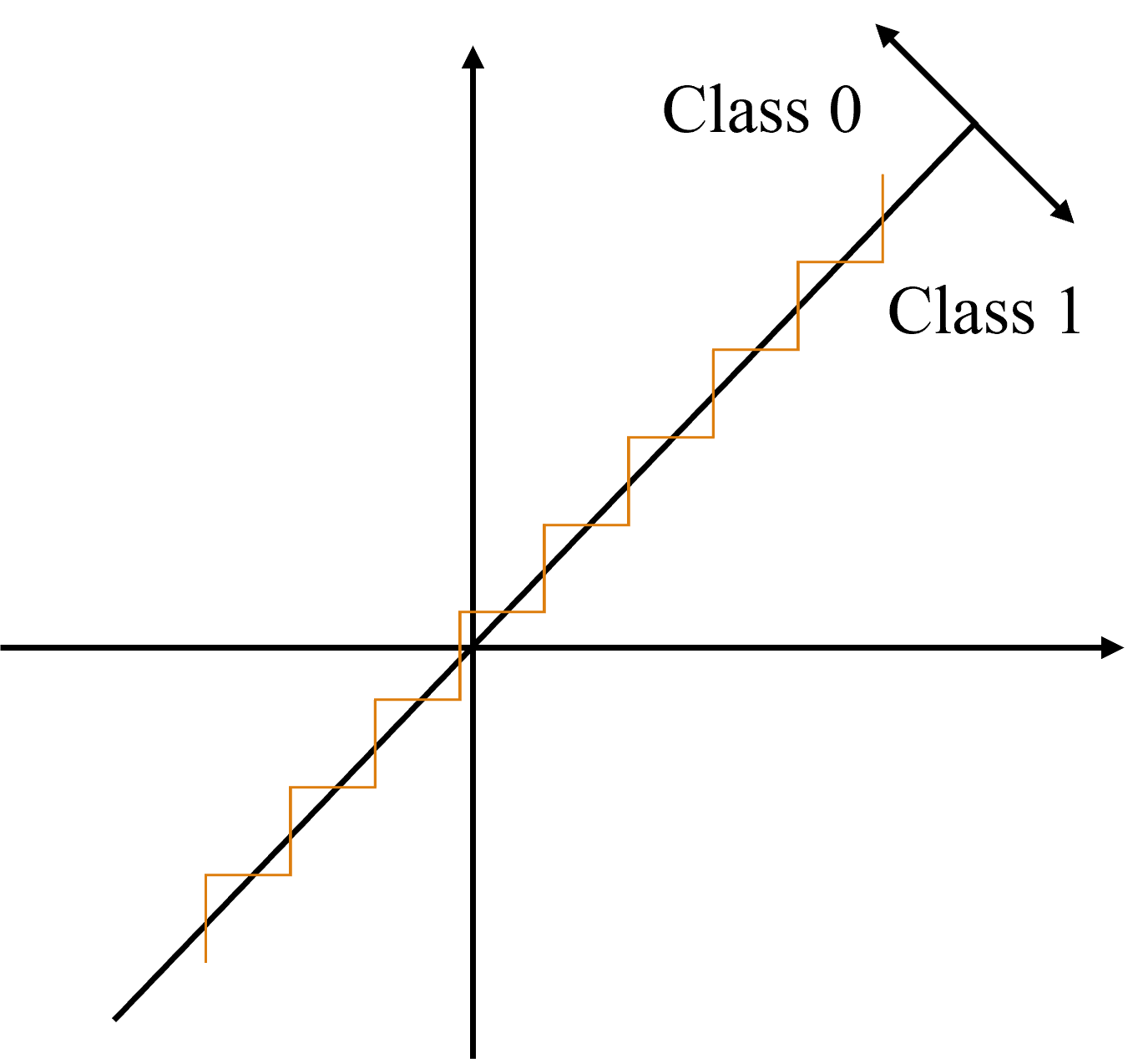}
    \caption{Example of Inefficient Decision Boundary of Decision Tree Model}
    \label{fig:decision_boundary}
\end{figure*}

\begin{figure*}[t] 
    \centering
    \includegraphics[width=\textwidth]{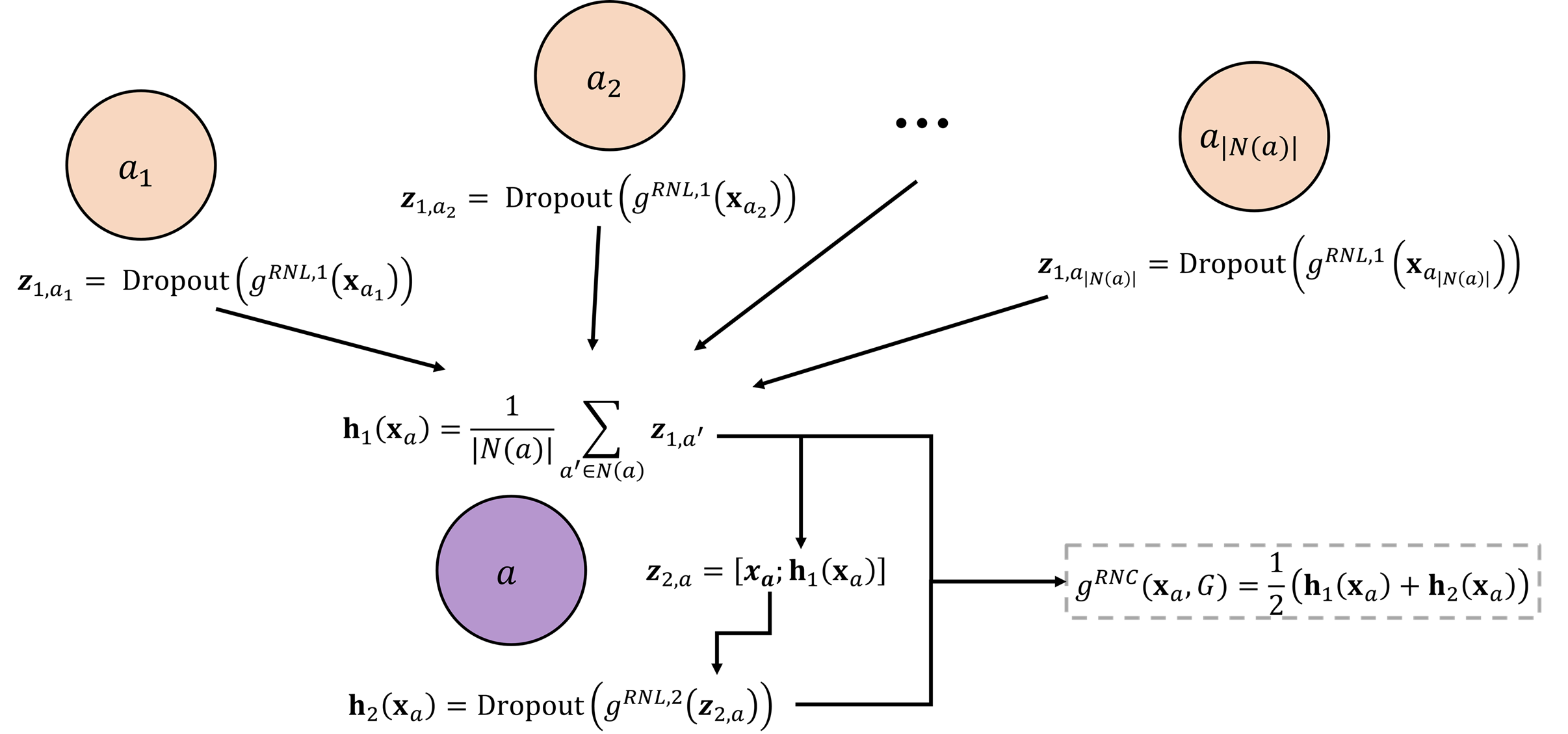}
    \caption{RNConv Layer}
    \label{fig:RNCLayer}
\end{figure*}

\textbf{RNConv Layer.}
Given a graph $G=(\mathcal{U},L)$\footnote{
Sets $\mathcal{U}$ and $L$ denote the node and edge sets, respectively, previously written as $(\mathcal{U}_{t}, L_{t-1}(p_{dg}))$. 
When discussing RNConv, we omit the subscripts $t$, $t-1$, and the argument $p_{dg}$, and simply write $G = (\mathcal{U}, L)$ instead of $G_{t-1}(p_{dg})=(\mathcal{U}_{t},L_{t-1}(p_{dg}))$ since RNConv can be applied to other problems.} and node feature vectors $\textbf{x}_a\in\mathbb{R}^p$ for all $a\in\mathcal{U}$, we define an RNConv layer as
\begin{equation}
    g^{RNC}(\textbf{x}_a,G) = \frac{1}{2} \left( \textbf{h}_1(\textbf{x}_a)+\textbf{h}_2(\textbf{x}_a) \right) \in \mathbb{R}^{n_{rdt}} \label{eq:rnc_layer}
\end{equation}
where
\begin{align}
    &\textbf{z}_{1, a} = \operatorname{Dropout} \left(g^{RNL,1}(\textbf{x}_{a});q_1 \right) \in \mathbb{R}^{n_{rdt}}, \label{eq:rnc_layer1} \\
    &\textbf{h}_1(x_a) = \frac{1}{|N(a)|}\sum_{a'\in N(a)} \textbf{z}_{1,a'} \in \mathbb{R}^{n_{rdt}}, \label{eq:rnc_layer2}\\
    &\textbf{z}_{2, a}= \left[
            \textbf{x}_a;
            \textbf{h}_1(\textbf{x}_a) \right] \in \mathbb{R}^{p+n_{rdt}}, \label{eq:rnc_layer3}\\
    &\textbf{h}_2(\textbf{x}_a) = 
        \operatorname{Dropout}\left(
            g^{RNL,2} \left( \textbf{z}_{2, a} \right);q_1 
        \right) \in \mathbb{R}^{n_{rdt}}, \label{eq:rnc_layer4}\\
    &N(a)=\{a': (a',a)\in L\} \label{eq:rnc_layer5}.
\end{align}
The architecture is visualized in \fref{fig:RNCLayer}. 
We first compute $\textbf{z}_{1,a}$'s for all $a\in\mathcal{U}$ in \eqref{eq:rnc_layer1}. 
By using Dropout, different subsets of trees in $g^{RNL,1}$ learn from different samples, enhancing tree diversity. 
Next, we aggregate the neighborhood's output by averaging $\textbf{z}_{1,a'}$ in \eqref{eq:rnc_layer2}. 
Then, we concatenate the input $\textbf{x}_a$ and aggregated vector $\textbf{h}_1(\textbf{x}_a)$ in \eqref{eq:rnc_layer3}, and pass the result into the second RNODE layer that extracts node-level information using both local and aggregated context in \eqref{eq:rnc_layer4}.
As in \eqref{eq:rnc_layer1}, we apply Dropout in \eqref{eq:rnc_layer4} to promote tree diversity. 
Finally, the output in \eqref{eq:rnc_layer} is the average of the neighborhood information  ($\textbf{h}_1$) and the node’s own representation ($\textbf{h}_2$).

\textbf{RNConv.} 
We define RNConv as 
\begin{equation}
    f^{RNC} (\textbf{x}_a,G) = \frac{1}{n_{rdt}l_{RNC}}\sum_{l=1}^{l_{RNC}} g^{RNC}_l(\textbf{e}_{l,a},G)\textbf{1} \in \mathbb{R} \label{eq:rnode} 
\end{equation} 
where
\begin{equation}
    \textbf{e}_{l,a} = \begin{cases}
    \operatorname{Dropout}(\textbf{x}_a;q_2) \in \mathbb{R}^p ~\text{if} ~l=0\\
    \left[\textbf{x}_a;\frac{1}{l-1}\sum_{i=1}^{l-1} g^{RNC}_{i}(\textbf{e}_{i,a})\right] \in \mathbb{R}^{p+n_{rdt}} ~o.w.
    \end{cases} \label{eq:rnode1}
\end{equation} 
Analogous to NODE in \eqref{eq:node}, we take the average over all RNODE layer outputs in \eqref{eq:rnode}.
In \eqref{eq:rnode1}, the average of the interim outputs is fed into the next NODE layer, instead of concatenation as in \eqref{eq:node1}, to avoid excessive parameter growth. 
Dropout is applied to enhance tree diversity.

\subsection{Synthetic-Long-Short Arbitrage}
\label{subsec:slsa}
In this section, we propose SLSA, a class of synthetic-long positions designed to tackle the problem of exploiting StatArbs using the predictions from \ssref{subsec:prediction}. 
We begin by presenting the synthetic-long-short position, which serves as the foundation for SLSA. 
We then propose SLSA and establish its minimal theoretical risk under the AF assumption. 
Lastly, we introduce the SLSA projection, which maps predictions into tradable SLSA positions capable of realizing StatArbs. 
Each position is executed at market open and held to maturity, as described in \sref{sec:method}, yielding profits upon inception without intermediate and final cash flows.

\textbf{Synthetic-Long-Short.}
We define a synthetic-long-short $(LS; U, \mathbf{n})$, as a synthetic asset consisting of $m\in\mathbb{R}$ units of the underlying instrument and $n_a \in \mathbb{R}$ contracts of $a \in U$, satisfying 
\begin{equation}
m+\sum_{a \in U} n_a = 0 \label{eq:def_ls} 
\end{equation}
where $U$ is a subset of $\mathcal{U}_t$, and $\mathbf{n}=[n_a]_{a\in U} \in \mathbb{R}^{|U|}$ is a row-stacked vector ordered lexicographically. 
Negative $n_a$ means a short position of $|n_a|$ contracts on $a$, which similarly applies to $m$. 
Since $m$ becomes automatically determined if $\textbf{n}$ is specified in \eqref{eq:def_ls}, $m$ is not required as a parameter. 
The time-$\tau \in \mathcal{T}(U)$ price of $(LS;U,\textbf{n})$ is
\begin{align}
&P_{\tau}(LS;U,\textbf{n}) \nonumber\\
&=mS_{\tau}+\sum_{a \in U} n_a P_{\tau}(SL;M_a,K_a) \label{eq:ls_price}\\
&=mS_{\tau}+\sum_{a \in U} n_a \left(S_{\tau}-K_a \delta_{a,\tau} \right) \label{eq:ls_price1}\\
&=\sum_{a \in U} -n_a K_a \left( y_{a,\tau}+\bar{\delta}_{M_a,\tau} \right) \label{eq:ls_price2}
\end{align}
where $\mathcal{T}(U)= \left\{\tau' : \{a\in U: n_a \neq 0\} \subset \mathcal{A}_{\lfloor \tau' \rfloor} (SL) \right\}$ represents the set of time points when all nonzero-position assets of $U$ are listed and unexpired. 
Its definition and \eqref{eq:unit_bond} yield \eqref{eq:ls_price} and \eqref{eq:ls_price1}, respectively.
We obtain \eqref{eq:ls_price2} from \eqref{eq:pred_target} and \eqref{eq:def_ls}. 
Since $P_{\tau}(LS;U,\textbf{n})$ includes both $y_{a,\tau}$ and $\bar{\delta}_{M_a,\tau}$, $(LS;U,\textbf{n})$ can contain both arbitrages and zero-coupon bonds, as discussed in Proposition \ref{prop:pred_target}.

Synthetic long-short positions possess desirable properties in that a position in $(LS; U, \mathbf{n})$ eliminates risks related to both the price and volatility of the underlying instrument and results in zero price variance, as demonstrated in Proposition~\ref{prop:ls}.
Although the two assumptions in Proposition~\ref{prop:ls} may not strictly hold in practice, several well-known pricing models---such as the \cite{black1973pricing} model and the Binomial pricing model~\citep{cox1979option}---are derived under these two assumptions. 
Thus, under the two assumptions, synthetic long-shorts can have equal or lower risk than positions based on such well-known models when they have the same time-related risk level.

\begin{proposition}
If there is no arbitrage in the market and the risk-free interest rate is constant, then the following statements hold for all $\tau \in \mathcal{T}(U)$:\\
\proman{1} each $(LS; U, \mathbf{n})$ is neutral with respect to both the price and volatility of its underlying instrument, \\
\proman{2} $Var(P_{\tau}(LS;U,\textbf{n}))=0$.
\label{prop:ls}
\end{proposition}
\begin{proof}{Proof}
Let $\tau$ be an arbitrary element of $\mathcal{T}(U)$. 
Then, the time-$\tau$ price of $(LS; U, \mathbf{n})$ becomes $P_{\tau}(LS;U,\textbf{n})=\sum_{a \in U} -n_a K_a \Pi_{\tau}(1;M_a)=\sum_{a \in U} -n_a K_a e^{-r_f(M_a-{\tau})}$, where the first equality follows from \eqref{eq:ls_price2} and Proposition~\ref{prop:pred_target}, and the second from \pref{eq:pv_constant}. 
Since $n_a, K_a, r_f, M_a$ are constants, $P_{\tau}(LS;U,\textbf{n})$ depends solely on $\tau$. 
Thus, the process $\{P_{\tau}(LS; U, \mathbf{n})\}_{\tau \in \mathcal{T}(U)}$ is stochastically independent of the volatility and price processes of the underlying instrument over $\mathcal{T}(U)$, implying that $P_{\tau}(LS;U,\textbf{n})$ is neutral to both.
As $P_{\tau}(LS;U,\textbf{n})$ is deterministic, $Var(P_{\tau}(LS;U,\textbf{n}))=0$. 
\Halmos
\end{proof}

However, even under the assumptions in Proposition~\ref{prop:ls}, $(LS; U, \mathbf{n})$ can still be exposed to time-related risk.  
Moreover, if the risk-free interest rate is not constant, as in \citet{vasicek1977equilibrium}, Proposition~\ref{prop:ls} no longer holds.  
Furthermore, synthetic-long-short positions can contain both a unit zero-coupon bond and arbitrage opportunities, but our goal is to exclusively exploit pure arbitrage opportunities.
To address these limitations, we propose \textbf{S}ynthetic-\textbf{L}ong-\textbf{S}hort-\textbf{A}rbitrage (SLSA).


\textbf{Synthetic-Long-Short-Arbitrage (SLSA)}
We define an SLSA $(SA; U, \mathbf{n})$ as a synthetic-long-short satisfying
\begin{align} 
&\sum_{a\in U: M_a=M} n_a=0, ~\forall M\in\mathcal{M}(U),  \label{eq:def_sa1}\\ 
& \sum_{a\in U: M_a=M} K_a \cdot n_a=0, ~\forall M\in\mathcal{M}(U) \label{eq:def_sa2}
\end{align}
where $\mathcal{M}(U) = \{M_a : a \in U\}$ denotes the set of maturities of all assets in $U$.
Because SLSA is a synthetic-long-short, its $\mathbf{n}$ also satisfies \eqref{eq:def_ls}. 
Due to \eqref{eq:def_ls} and \eqref{eq:def_sa1}, we have $m=0$ for every SLSA. 
For all $\tau\in\mathcal{T}(U)$, the time-$\tau$ price of a SLSA is 
\begin{align}
&P_\tau(SA;U,\textbf{n})\nonumber\\
&=\sum_{M\in\mathcal{M}(U)} \left( - \bar{\delta}_{M,\tau} \sum_{a\in U:M_a=M} n_a K_a\right) - \sum_{a \in U} n_a K_a y_{a,\tau} \label{eq:slsa_price1}  \\
&=\sum_{a \in U} -n_a K_a y_{a,\tau}, \label{eq:slsa_price2}
\end{align}
where from \eqref{eq:ls_price2} and the definition of $\mathcal{M}(U)$, \eqref{eq:slsa_price1} follows.  
Its first term vanishes by \eqref{eq:def_sa2}, yielding \eqref{eq:slsa_price2}.

For $t\in\mathbb{N}$, if each component option position of $(SA; U, \textbf{n})$ is built at $o(t)$ and held until its respective expiration date, as outlined  in the introduction of Section~\ref{sec:method}, then the payoff is as follows:
\begin{equation}
    \text{Payoff}_{\tau}{(SA;U,\textbf{n})}=\begin{cases}
    \sum_{a \in U} n_a v_{a,t} &\text{if}~\tau=o(t)\\
    0 &\text{if}~\tau \neq o(t)
    \end{cases}
    \label{eq:payoff_sa1}
\end{equation}
where $v_{a,t}=K_a y_{a,o(t)}$. 
At $o(t)$, we pay $P_{o(t)}(SA; U, \textbf{n})$ for the component options, as derived in \eqref{eq:slsa_price2}.  
For each $M \in \mathcal{M}(U)$, only the options $a \in \{a' \in U : M_{a'} = M\}$ maturing at $M$ yield payoff. 
Thus, the payoff at $M$ is  $\sum_{a \in U : M_a = M} n_a (S_M - K) = (S_M - K) \sum_{a \in U : M_a = M} n_a = 0$, by \eqref{eq:sl_payoff} and \eqref{eq:def_sa1}.  
No payoff occurs at the other time points.  
Thus, the only nonzero net cash flow of $(SA, U, \textbf{n})$ is $-P_{o(t)}(SA; U, \textbf{n})$ at the initial time point $o(t)$. 
Accordingly, allocating large $n_a$ to $a$ with large $v_{a,t}$ can increase profit.

Compared to the synthetic-long-shorts that can contain bonds as shown in \eqref{eq:ls_price2}, SLSA solely comprises arbitrage opportunities since $P_\tau(SA; U, \mathbf{n})$ is a linear combination of arbitrages $y_{a,\tau}$ with the fixed coefficients of $-n_a K_a$ in \eqref{eq:slsa_price2}.
Moreover, if the market is AF, the price of $(SA; U, \mathbf{n})$ vanishes, as shown by the proposition below. 
Hence, we confirm arbitrage opportunities are exclusively contained in SLSA. 

\begin{proposition}
\label{prop:slsa}
Under the AF assumption, the following statements hold for all $\tau \in \mathcal{T}(U)$: \\
\proman{1} $P_\tau(SA;U,\textbf{n})=0$, \\
\proman{2} $Var(P_\tau(SA;U,\textbf{n}))=0$, \\
\proman{3} every SLSA is neutral with respect to all \cite{black1973pricing} risk factors (i.e., the risk-free interest rate, time to maturity, and the price and volatility of its underlying instrument).
\end{proposition}
\begin{proof}{Proof}
Proposition \ref{prop:pred_target} and  \eqref{eq:slsa_price2} yield $P_\tau(SA;U,\textbf{n})=0$ for all $\tau\in\mathcal{T}(U)$.
Thus, process $\{P_{\tau}(SA; U, \mathbf{n})\}_{\tau \in \mathcal{T}(U)}$ is stochastically independent of all \cite{black1973pricing} risk factors. 
Hence, every SLSA is neutral to all \cite{black1973pricing} risk factors.  
As $P_\tau(SA;U,\textbf{n})$ is a constant, its variance is zero for all $\tau\in\mathcal{T}(U)$. \Halmos
\end{proof}

SLSA exhibits the lowest risk level among the trading strategies developed in the past under the AF assumption since SLSA has zero variance and is neutral with respect to all the \cite{black1973pricing} risk factors, as shown in Proposition \ref{prop:slsa}.
Consequently, SLSA solves the problem of exposure to time in Proposition \ref{prop:ls}. 
In addition, Proposition \ref{prop:slsa} does not assume a constant risk-free interest rate, compared to Proposition \ref{prop:ls}. 
SLSA resolves the limitations of the synthetic-long-short. 

Although the AF assumption in Proposition \ref{prop:slsa} may fail to hold in reality, the AF assumption is fundamental to rational decision-making in economics \citep{nau1991arbitrage}.
Moreover, as arbitrage precludes the existence of a present-value function \citep{skiadas2024theoretical}, trading strategies derived in the extensive literature employing the present-value function assume an AF market. 
While real-world markets may contain arbitrages, they are expected to gravitate toward an AF state \citep{varian1987arbitrage,langenohl2018sources}, rendering SLSA relatively low-risk.

In light of the desirable properties established in Proposition \ref{prop:slsa}, we adopt SLSA for our trading positions.
We next detail the methodology for determining $\textbf{n}$ in $(SA; U, \textbf{n})$.

\textbf{Synthetic-Long-Short-Arbitrage Projection.}
For $t\in\mathbb{N}$, given the predicted values $\hat{y}_{a,o(t)}$ for $y_{a,o(t)}$ defined in~\eqref{eq:pred_target}, we construct our position of a SLSA $(SA;\mathcal{U}_t,\textbf{n}_{t})$ as follows:
\begin{equation}
\textbf{n}_t = \text{Proj}_{\text{Null}(A_t)}\hat{\textbf{v}}_{t}= N_t(N_t^TN_t)^{-1}N_t^T \hat{\textbf{v}}_{t} \label{eq:proj_slsa1}
\end{equation}  
where $\textbf{n}_t =[n_{a,t}]_{a\in \mathcal{U}_t} \in \mathbb{R}^{|\mathcal{U}_t|}$.
The matrix $A_t$ represents the constraints in~\eqref{eq:def_sa1} and~\eqref{eq:def_sa2}, and $N_t$ consists of orthonormal columns spanning $\text{Null}(A_t)$. 
Then, the projection matrix is $N_t(N_t^TN_t)^{-1}N_t^T$. 
A row-stacked vector $\hat{\textbf{v}}_{t} = [\hat{v}_{a,t}]_{a \in \mathcal{U}_t} \in \mathbb{R}^{|\mathcal{U}_t|}$ consists of predicted values $\hat{v}_{a,t} \in \mathbb{R}$ for $v_{a,t}$, computed as $\hat{v}_{a,t} = K_a \hat{y}_{a,o(t)}$, where $\hat{y}_{a,o(t)} \in \mathbb{R}$ is the prediction for $y_{a,o(t)}$ defined in~\eqref{eq:pred_target}.

To construct an SLSA position satisfying \eqref{eq:def_sa1} and~\eqref{eq:def_sa2} while maximizing expected profit with consideration of prediction risk, we formulate our position as in \eqref{eq:proj_slsa1}.
Since $\textbf{n}_t$ is the projection onto $\text{Null}(A_t)$, \eqref{eq:def_sa1} and~\eqref{eq:def_sa2} are automatically satisfied. 
As shown in \eqref{eq:payoff_sa1}, profit increases with greater allocation to assets $a$ with larger $v_{a,t}$.  
We thus define the initial (non-SLSA) position as $\hat{\textbf{v}}_{t}$, assigning positions in proportion to the predicted values $\hat{v}_{a,t}$. 
Since these positions are based on predicted values $\hat{v}_{a,t}$ rather than the ground-truth $v_{a,t}$, we avoid concentrating our positions in assets with the most extreme predicted payoffs.  
That is, we mitigate prediction risk arising from prediction errors by avoiding concentration in such extreme assets.
Lastly, since projection matrices are positive semidefinite, \eqref{eq:proj_slsa1} ensures a non-negative aggregate predicted payoff, i.e., $\hat{\textbf{v}}^T_{t}\textbf{n}_t \geq 0$.

\section{Experiments}
\label{sec:experiments}
In this section, we show that RNConv statistically significantly outperforms benchmarks, and that the P\&L of SLSA projection positions exhibits a consistently upward trend. 
We begin by delineating the experimental setup, followed by the results of universe selection, RNConv, and SLSA.

\subsection{Experimental Setup}
We utilized historical data on KOSPI 200 and its options spanning from July 7th, 1997 to December 30th, 2024, which are available at \cite{krx_data_2025,krx_global_2025}.
The data include detailed option information (e.g., option type, strike price) as well as historical time-series data of the close, open, high, and low prices for each option and for the KOSPI 200 index itself.
We selected the KOSPI 200 options primarily because the Korea Exchange---a reliable data provider---offers relatively long-term historical data free of charge.
Following the data splits in \ssref{subsec:setup_notation},
we set $t_{fit,1}$ as the first trading date of 2015 to ensure a few thousand training graphs, and we defined $\mathcal{T}_{fit}=\{t_{fit,1}, \dots, t_{fit,n_{fit}}\}$ as the first trading dates of each quarter (i.e., $\vert \mathcal{T}_{fit} \vert = 40$), and set $p_{val} = 0.2$.

We implemented our method and conducted the experiments in \texttt{Python} and \texttt{R}, using \texttt{PyTorch}, \texttt{PyTorch-Geometric}, \texttt{scikit-learn}, and \texttt{lawstat} libraries, with the following  values for the hyperparameters.  
We set $p_{dg} = 1/3$ so that, with at least three GNN layers, each node can aggregate information from all others. 
In \eqref{eq:rdodt4}, $VBN(\cdot)$ is implemented in \texttt{PyTorch} with \texttt{affine=False}. 
In \pref{eq:rdodt4}, we set $\alpha = 1.5$  following~\citep{popov2019neural}, and $\gamma = 5$ to approximate a 20\% probability that $\sigma_\alpha(Z) \in [0.1, 0.9]$, assuming $Z \sim \mathcal{N}(0,1)$ and applying Pareto's rule. 
In \pref{eq:rdodt5}, $\text{MLP}(\cdot)$ has three layers, where the first two hidden layers have $p/4$ neurons, as in $CrossNet$. 
Leaky ReLU is used as the activation in both $\text{MLP}(\cdot)$ and $\text{CrossNet}(\cdot)$.
We set $l_{CN}=2$ in \pref{eq:crossnet_rnode} and $p_{CN}=p/4$ in \pref{eq:crossnet2}  and \pref{eq:original_dcn}, following~\citet{wang2021dcn}. 
Finally, we set $q_1=0.5$ in \eqref{eq:rnc_layer1} for the interim layers and  $q_2 = 0.2$ in \eqref{eq:rnode1} for the input layer, following \citet{srivastava2014dropout}.








\subsection{Trading Universe Selection Model}
\label{subsec:exp_universe}

We employed a radius neighbor classifier to predict $\hat{\mu}_{tr,a,t}$ in \eqref{univ:obj}--\eqref{univ:end}, using \texttt{scikit-learn} and scaled\footnote{Each feature was first transformed to follow a uniform distribution via the probability integral transformation, and then mapped to a standard normal distribution using the inverse cumulative distribution. The same scaling method was applied to the arbitrage prediction inputs and outputs.} features of moneyness ($S_{t-1,c}/K_a$) and time to maturity ($M_a - t$) in $\textbf{x}_{tr,a,t-1}$.
For each $\mathcal{D}_i^{tr}$ from \sref{subsec:setup_notation}, 50 radius neighbor classifiers were trained with radii linearly spaced between 0.0 and 0.1 (excluding 0.0). 
The best model $\hat{f}^{tr}_{i}$ was selected based on $\mathcal{D}_{val,i}^{tr}$ and predicted $\hat{\mu}_{tr,a,t}$ on $\mathcal{D}_{test,i}^{tr}$. 
If no neighbors \edit{exist} within a radius, 5-nearest neighbors \edit{are} used. 
We also evaluated the best model $\hat{f}^{tr}_{i}$ on $\mathcal{D}_{test,i}^{tr}$, which are summarized in \tref{tab:tradability_stats}. 

We solved \eqref{univ:obj}--\eqref{univ:end} for $p_{univ} \in \{16, 24, 32\}$ utilizing \texttt{PuLP} and \texttt{Gurobi}. 
When infeasible, the largest feasible $p_{univ}$ was selected. 
For example, if \eqref{univ:obj}--\eqref{univ:end} \edit{are} infeasible for $p_{univ} = 24$, we gradually \edit{reduce} $p_{univ}$ (i.e., 23, 22, ...) until a feasible solution \edit{is} found. 
The resulting universe cardinalities over time are shown in \fref{fig:universe_time_series}. 
The KOSPI 200 options market was inaugurated in 1997. 
In its early years, option trading increased gradually, resulting in the upward trend in \fref{fig:universe_time_series}.
Subsequently, a sufficient number of options were traded for $p_{univ} \in \{16, 24\}$, while trading volume remained insufficient for $p_{univ} = 32$, as indicated by the fluctuations observed from 2016 to 2024.

\begin{table}[thbp]
\centering
\caption{Descriptive Statistics of Model Evaluation Metrics for Tradability Prediction (2015 Q1–2024 Q4)}
\label{tab:tradability_stats}
\begin{tabular}{crrrrrr}
\toprule
 & Accuracy & Precision & Recall & F1 & ROC AUC \\
\midrule
Mean & 76.939 & 77.597 & 97.867 & 86.531 & 83.803 \\
Std & 2.842 & 2.830 & 1.538 & 1.916 & 2.699 \\
Min & 71.601 & 71.424 & 93.401 & 82.890 & 78.602 \\
Max & 85.208 & 85.359 & 99.767 & 91.990 & 89.370 \\
\bottomrule
\end{tabular}
\end{table}


\begin{figure*}[thbp] 
    \centering
    \includegraphics[width=1.0\textwidth]{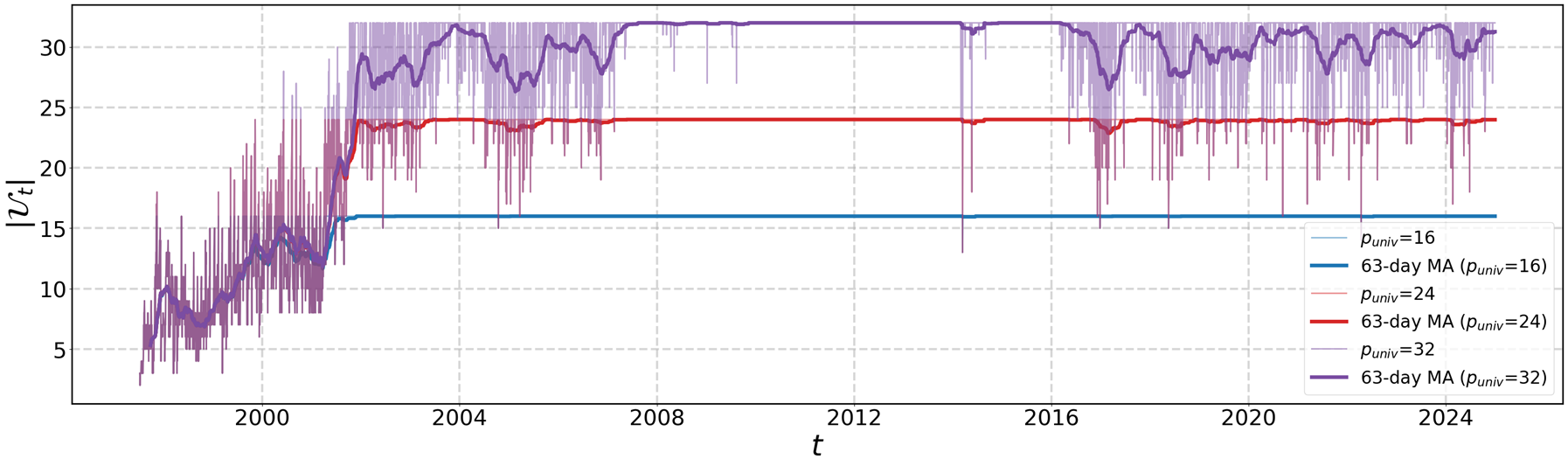}
    \caption{Time Series of $\vert \mathcal{U}_t \vert$. \textnormal{MA stands for the moving average.}}
    \label{fig:universe_time_series}
\end{figure*}

\subsection{Arbitrage Prediction}
\label{subsec:arb_pred}

\textbf{Benchmarks.} 
Given a node feature $\textbf{x}\in\mathbb{R}^p$ and graph $G$, we considered the following architecture for the benchmark convolution methods to evaluate the performance of RNConv:
\begin{equation}
    f^{BM}(\textbf{x},G;\theta) = \operatorname{Head}(\textbf{z}_{l_{BM}}) \in \mathbb{R} \label{eq:bm1}
\end{equation}
where for $l=1,2,...,l_{BM},$
\begin{equation}    \textbf{z}_l=\operatorname{LeakyReLU}\left(g^{BM}_l(\textbf{z}_{l-1},G;\theta_l)\right). \label{eq:bm2}
\end{equation}
Function Head is a linear layer without activation, and $g_l^{\mathrm{BM}}$ denotes a benchmark graph convolution with parameters $\theta_l$.  
We used GCN, GAT, SAGE, and GPS as instances of $g_l^{\mathrm{BM}}$, implemented using \texttt{GCNConv}, \texttt{GATConv}, \texttt{SAGEConv}, and \texttt{GPSConv} from the \texttt{PyTorch-Geometric} library.\footnote{
Additionally, parametric models (e.g., \cite{black1973pricing}, \cite{heston1993closed}) based on the AF assumption or present-value functions are unsuitable as benchmarks, since they predict $y_{a,t} = 0$ for all $a$ and $t$, as shown in Proposition~\ref{prop:pred_target}.
}

Since deep learning model performance depends on the number of layers and parameters, we performed a grid search over these hyperparameters using $\mathcal{D}^{ar}_{train,i},\mathcal{D}^{ar}_{val,i}$ in each period indexed by $t_i\in\mathcal{T}_{fit}$ for each $p_{univ}$ and for each graph convolution architecture. 
In this way, we obtained the best model $\hat{f}_i^{m}$ for each $i,p_{univ},m$. 
Lastly, we computed $\hat{y}_{a,t}$ and the mean squared error (MSE) of $\hat{f}_i^{m}$ on  $\mathcal{D}^{ar}_{test,i}$ for each $i,p_{univ},m$. 

We considered the following hyperparameter grids for $l_{BM}$ and the number of parameters: \edit{$\{3,4,5\}\times\{10^5,~5\cdot10^5,~10^6\}$.} 
However, as setting the number of parameters as \edit{$10^5$,~$5\cdot10^5$, or $10^6$} is infeasible in many cases, we considered the following to set the number of hyperparameters close to these numbers: $p_{ctr} \in \arg\min_{p'_{ctr}} \left| \#\text{params}(f^{m}(\textbf{x}, G; \theta(p'_{ctr}))) - n^* \right|$
where $n^*$ is the target number of parameters (e.g., \edit{$10^5$}), and $p_{ctr}$ is a hyperparameter controlling the number of the parameters of $\theta(p_{ctr})$.\footnote{ 
For RNC, $p_{ctr}$ is the number of trees of an RNC layer, which is the same for all the layers.  
For GCN, SAGE, GAT, and GPS, $p_{ctr}$ is the number of their hidden neurons---specifically, \texttt{out\_channels} and \texttt{channels} in \texttt{Pytorch-Geometric}. 
The other hyperparameters were set as the values suggested in the underlying articles.}

\textbf{Experimental Results.} 
Model performance was evaluated using the MSEs on $\mathcal{D}^{ar}_{\text{test},i}$ for all $i$. 
As shown in \tref{tab:avg_mse}, our method consistently outperforms the benchmarks in terms of average MSEs across all $p_{univ}$ values.  
The MSE magnitudes are low enough to yield profits, which we discuss in \sref{subsec:exp_arbitrage_strategies}. 
Due to the nature of the prediction target, the MSEs are significantly influenced by unobserved external factors: \fref{fig:time_series_avg_mse} shows the average with respect to $p_{univ}$. The zoom-in box shows the values for all five methods. 
For example, the early 2020 spike corresponds to market disruptions from COVID-19, whose unprecedented dynamics~\citep{baker2020unprecedented} were not reflected in the training data or features. 

We conducted statistical tests on the MSE differences between our method and each benchmark, similarly to \cite{wu2021machine}. 
Since the average MSEs in \tref{tab:avg_mse} might be distorted by large values as illustrated in \fref{fig:time_series_avg_mse},  
\tref{tab:hypotheses} summarizes several tests applied to the differences $\hat{\sigma}_{bm}^2-\hat{\sigma}_{RNC}^2$ between the MSEs of RNConv and a benchmark for the same $i,p_{univ}$ pairs. 

For all values of $p_{univ}$, our method statistically significantly outperforms all benchmarks, as summarized in \tref{tab:p_values_metrics}.  
Bold values indicate p-values for which the null hypothesis in the difference comparison is rejected under a significance level of 0.05, provided that the assumptions of the corresponding statistical tests are satisfied.

\begin{table}[thbp]
\centering
\caption{Averages MSEs over time (in $10^{-6}$)}
\label{tab:avg_mse}
\begin{tabular}{lrrrrr}
\toprule
 $p_{univ}$ & GAT & GCN & GPS & SAGE & RNConv \\
\midrule
16 & 6.3618 & 6.3361 & 6.3748 & 6.3343 & \textbf{6.3254} \\
24 & 8.5003 & 8.4995 & 8.5042 & 8.5027 & \textbf{8.4767} \\
32 & 9.5219 & 9.5234 & 9.5407 & 9.5440 & \textbf{9.5034} \\
\bottomrule
\end{tabular}
\end{table}

\begin{figure*}[thbp] 
    \centering
    \includegraphics[width=1.0\textwidth]{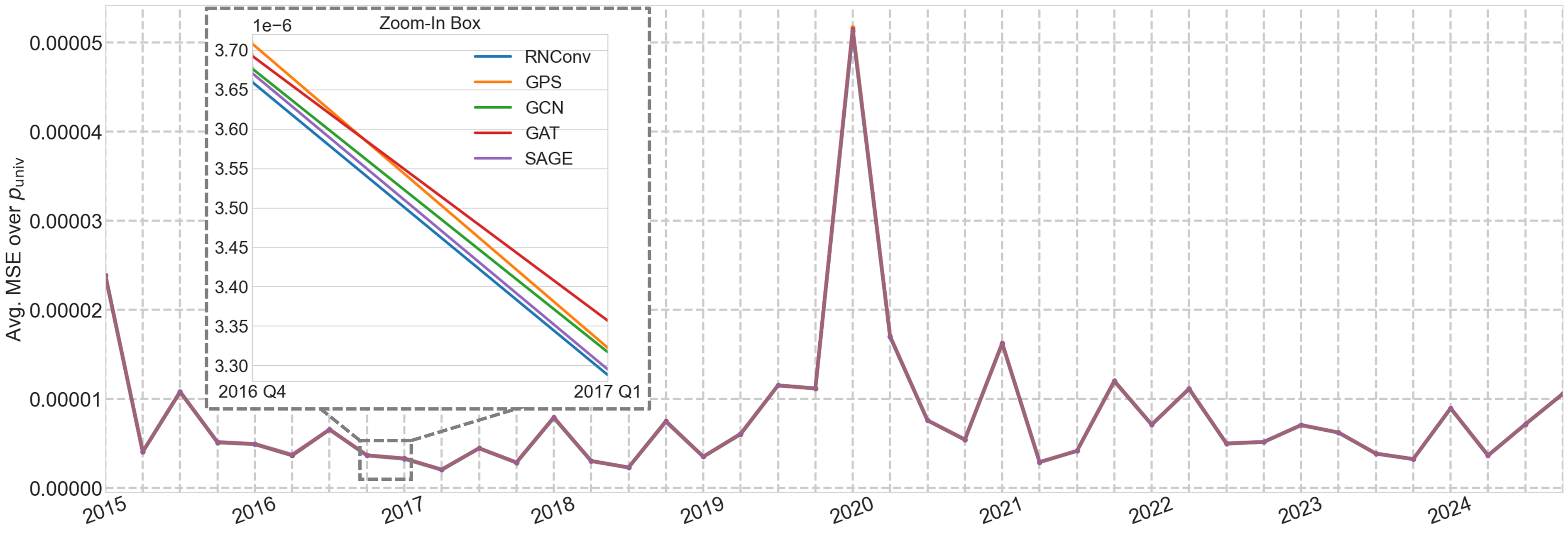}
    \caption{Time Series of Avg. MSE over $p_{univ}$ \edit{for RNConv}}
    \label{fig:time_series_avg_mse}
\end{figure*}

\begin{table}[htbp]
\centering
\caption{Hypotheses of \tref{tab:p_values_metrics}. \textnormal{$\hat{\sigma}_{bm}^2$ and $\hat{\sigma}_{RNC}^2$} denote the MSEs of a benchmark and RNConv, respectively.}
\label{tab:hypotheses}
\begin{tabularx}{\linewidth}{@{}lX X@{}}
\toprule
\textbf{Test} & \textbf{Null Hypothesis ($H_0$)} & \textbf{Alternative Hypothesis ($H_1$)} \\
\midrule
Paired T & $\text{mean}(\hat{\sigma}^2_{\text{bm}} - \hat{\sigma}^2_{\text{RNC}}) \leq 0$ & $\text{mean}(\hat{\sigma}^2_{\text{bm}} - \hat{\sigma}^2_{\text{RNC}}) > 0$ \\
Wilcoxon Signed-Rank  & $\text{median}(\hat{\sigma}^2_{\text{bm}} - \hat{\sigma}^2_{\text{RNC}}) \leq 0$ & $\text{median}(\hat{\sigma}^2_{\text{bm}} - \hat{\sigma}^2_{\text{RNC}}) > 0$ \\
Permutation  & $\text{mean}(\hat{\sigma}^2_{\text{bm}} - \hat{\sigma}^2_{\text{RNC}}) = 0$ & $\text{mean}(\hat{\sigma}^2_{\text{bm}} - \hat{\sigma}^2_{\text{RNC}}) > 0$ \\
Shapiro-Wilk  & $\hat{\sigma}^2_{\text{bm}} - \hat{\sigma}^2_{\text{RNC}} \sim \mathcal{N}(\mu, \sigma^2)$ & $\hat{\sigma}^2_{\text{bm}} - \hat{\sigma}^2_{\text{RNC}} \not\sim \mathcal{N}(\mu, \sigma^2)$ \\
Kolmogorov-Smirnov  & $\hat{\sigma}^2_{\text{bm}} - \hat{\sigma}^2_{\text{RNC}} \sim \mathcal{N}(\hat{\mu}, \hat{\sigma}^2)$ & $\hat{\sigma}^2_{\text{bm}} - \hat{\sigma}^2_{\text{RNC}} \not\sim \mathcal{N}(\hat{\mu}, \hat{\sigma}^2)$ \\
\makecell[l]{Symmetry\\\citep{miao2006new}} & $\hat{\sigma}^2_{\text{bm}} - \hat{\sigma}^2_{\text{RNC}}$ is symmetric 
& $\hat{\sigma}^2_{\text{bm}} - \hat{\sigma}^2_{\text{RNC}}$ is not symmetric 
\\
\bottomrule
\end{tabularx}
\end{table}
\begin{table}[thbp]
\centering
\caption{P-values from Statistical Tests on MSE Differences 
}
\label{tab:p_values_metrics}
\begin{tabular}{llcccccc}
\toprule
$p_{univ}$ & Benchmark & Paired T & \makecell{Wilcoxon\\Signed-Rank} & Permutation & \makecell{Shapiro-\\Wilk} & \makecell{Kolmogorov-\\Smirnov} & Symmetry \\
\midrule
\multirow[t]{4}{*}{16} & GAT & 0.0659 & \textbf{0.0001} & 0.0001 & 0.0000 & 0.0000 & 0.1462 \\
 & GCN & \textbf{0.0009} & 0.0013 & 0.0009 & 0.9718 & 0.8945 & 0.9642 \\
 & GPS & 0.1278 & 0.0809 & \textbf{0.0244} & 0.0000 & 0.0000 & 0.0352 \\
 & SAGE & \textbf{0.0196} & 0.0093 & 0.0195 & 0.2657 & 0.9584 & 0.6572 \\
\cline{1-8}
\multirow[t]{4}{*}{24} & GAT & 0.0001 & \textbf{0.0000} & 0.0001 & 0.0116 & 0.2810 & 0.1244 \\
 & GCN & 0.0005 & \textbf{0.0000} & 0.0005 & 0.0113 & 0.3573 & 0.5212 \\
 & GPS & 0.0073 & 0.0030 & \textbf{0.0065} & 0.0000 & 0.0344 & 0.0132 \\
 & SAGE & 0.0003 & 0.0000 & \textbf{0.0001} & 0.0010 & 0.0813 & 0.0244 \\
\cline{1-8}
\multirow[t]{4}{*}{32} & GAT & \textbf{0.0030} & 0.0028 & 0.0027 & 0.1473 & 0.5161 & 0.4722 \\
 & GCN & \textbf{0.0013} & 0.0006 & 0.0010 & 0.0821 & 0.5675 & 0.7320 \\
 & GPS & 0.0001 & 0.0000 & \textbf{0.0001} & 0.0017 & 0.1606 & 0.0372 \\
 & SAGE & 0.0000 & 0.0000 & \textbf{0.0000} & 0.0000 & 0.0519 & 0.0134 \\
\bottomrule
\end{tabular}
\end{table}

\subsection{Synthetic-Long-Short Arbitrage}
\label{subsec:exp_arbitrage_strategies}

\textbf{Benchmarks.}
To evaluate the performance of our method, projected SLSA $(SA;\mathcal{U}_t,\textbf{n}_{t})$, we consider the following two benchmark strategies:
\begin{itemize}
    \item Benchmark 1: positions satisfying \eqref{eq:def_sa1}, 
    \item Benchmark 2: positions satisfying \eqref{eq:def_ls} with $m=0$,
\end{itemize}
where, in both cases, the positions are the projections of $\hat{\textbf{v}}_t \in \mathbb{R}^{|\mathcal{U}_t|}$ onto their respective constraint sets, analogous to the SLSA projection.
Unlike SLSA, Benchmark 1 may violate \eqref{eq:def_sa2}, and Benchmark 2 may fail to satisfy both \eqref{eq:def_sa1} and \eqref{eq:def_sa2}, thus functioning as relaxed ablation variants.

Since our prediction target is newly proposed in this paper and therefore lacks an established benchmark strategy for position-taking, we ad hoc adapted the strategies of \cite{wang2024considering} and \cite{wang2024deep} as Benchmark 1 and 2, respectively.
In evaluating prediction models for individual put and call option pricing, \cite{wang2024considering} considered symmetric-long-short positions within the same-maturity options, while \cite{wang2024deep} considered them for all options.  
Equation \eqref{eq:def_sa1} implements the former, whereas \eqref{eq:def_ls} with $m = 0$ realizes the latter.
Here, symmetric-long-short positions refer to equal numbers of long and short contracts.

Compared to SLSA's payoff \eqref{eq:payoff_sa1}, the benchmarks can have nonzero payoffs at maturity $M\in\mathcal{M}(\mathcal{U}_t) = \{M_a : a \in \mathcal{U}_t\}$. 
Their net payoffs at $M\in\mathcal{M}(\mathcal{U}_t)$ are as follows: 
\begin{align}
    \text{Payoff}_M(BM_1; \mathcal{U}_t, \textbf{n}) &= \sum_{a \in \mathcal{U}_t : M_a = M} -n_a K_a \label{eq:maturity_payoff_bm1} \\
    \text{Payoff}_M(BM_2; \mathcal{U}_t, \textbf{n}) &= \sum_{a \in \mathcal{U}_t : M_a = M} n_a (S_M - K_a) \label{eq:maturity_payoff_bm2}
\end{align}
while both of their inception payoffs are 
\begin{equation}
    \sum_{a \in \mathcal{U}_t} n_a K_a \left( y_{a,o(t)}+\bar{\delta}_{M_a,o(t)} \right) \label{eq:inception_payoff_bm2}
\end{equation}
from \eqref{eq:ls_price2}, and all the payoffs at the other time points are zero. 
Both \eqref{eq:maturity_payoff_bm1} and \eqref{eq:maturity_payoff_bm2} are derived from \eqref{eq:sl_payoff}. 
Since Benchmark 1 satisfies \eqref{eq:def_sa1}, the terms of $S_M$ cancel out.

Compared to SLSA, both benchmarks can contain synthetic bonds, and Benchmark 1  can additionally hold the underlying instrument. 
In \eqref{eq:maturity_payoff_bm1} and \eqref{eq:maturity_payoff_bm2}, nonzero $\sum_{a \in \mathcal{U}_t:M_a=M} n_a K_a$, indicates exposure to synthetic bonds. 
Nonzero $\sum_{a \in \mathcal{U}_t : M_a = M} n_a S_M$ in \eqref{eq:maturity_payoff_bm2} implies a position in the underlying instrument.

The RNConv predictions $\hat{y}_{a,o(t)}$ are used to compute $\hat{v}_{a,t}$, from which the projected SLSA and benchmark positions are derived.  
Each morning, positions totaling one long and one short contract are constructed and held until each option's expiration.
Performance are evaluated across $p_{univ}$.



\textbf{Experimental Results.} 
The cumulative sums of the profits and losses (P\&L) of SLSA show steady upward slopes for all $p_{univ}$, as shown in \fref{fig:slsa_pnl}, of which evaluation metrics can be found in \tref{tab:slsa_summary}. 
In \fref{fig:slsa_pnl}, transaction costs are set at 0.09\% of traded monetary values, reflecting the minimum commission rate offered by Miraeasset Securities in Korea. 
In \tref{tab:slsa_summary}, treating $\text{P\&L} / (\text{Number~Of~Contracts})$ as the return, information and Sortino ratios are computed with a zero benchmark return.
HHI denotes the Herfindahl-Hirschman Index, a measure of position concentration. 
The effective N is the reciprocal of the HHI and reflects how many options are effectively held, based on the distribution of positions.
The values of $avg \vert S-K\vert$ and $avg \vert M-t \vert$ represent the averages of $\sum_{a\in\mathcal{U}_t} n_{a,t} \vert S_t - K_a \vert$ and $\sum_{a\in\mathcal{U}_t} n_{a,t} \vert M_a - t \vert$, respectively.
These indicate how much the strategy allocates to deeper ITM or OTM options and to longer-term maturities. 
Since larger $p_{univ}$ invests in deeper-OTM-ITM options, as shown in $avg \vert S-K\vert$ in \tref{tab:slsa_summary}, and since such options can provide larger arbitrages, as shown in \fref{fig:moneyness_arbitrages}, larger-$p_{univ}$ projected SLSA positions can yield more profits. 
The projected SLSA generates small P\&Ls in 2015, which is, in our opinion, because the projected SLSA position is relatively closer to the signal having relatively high errors in 2015, as shown in Figures \ref{fig:time_series_avg_mse} and \ref{fig:cosine_similarity}. 

\begin{figure*}[thbp] 
    \centering
    \includegraphics[width=1.0\textwidth]{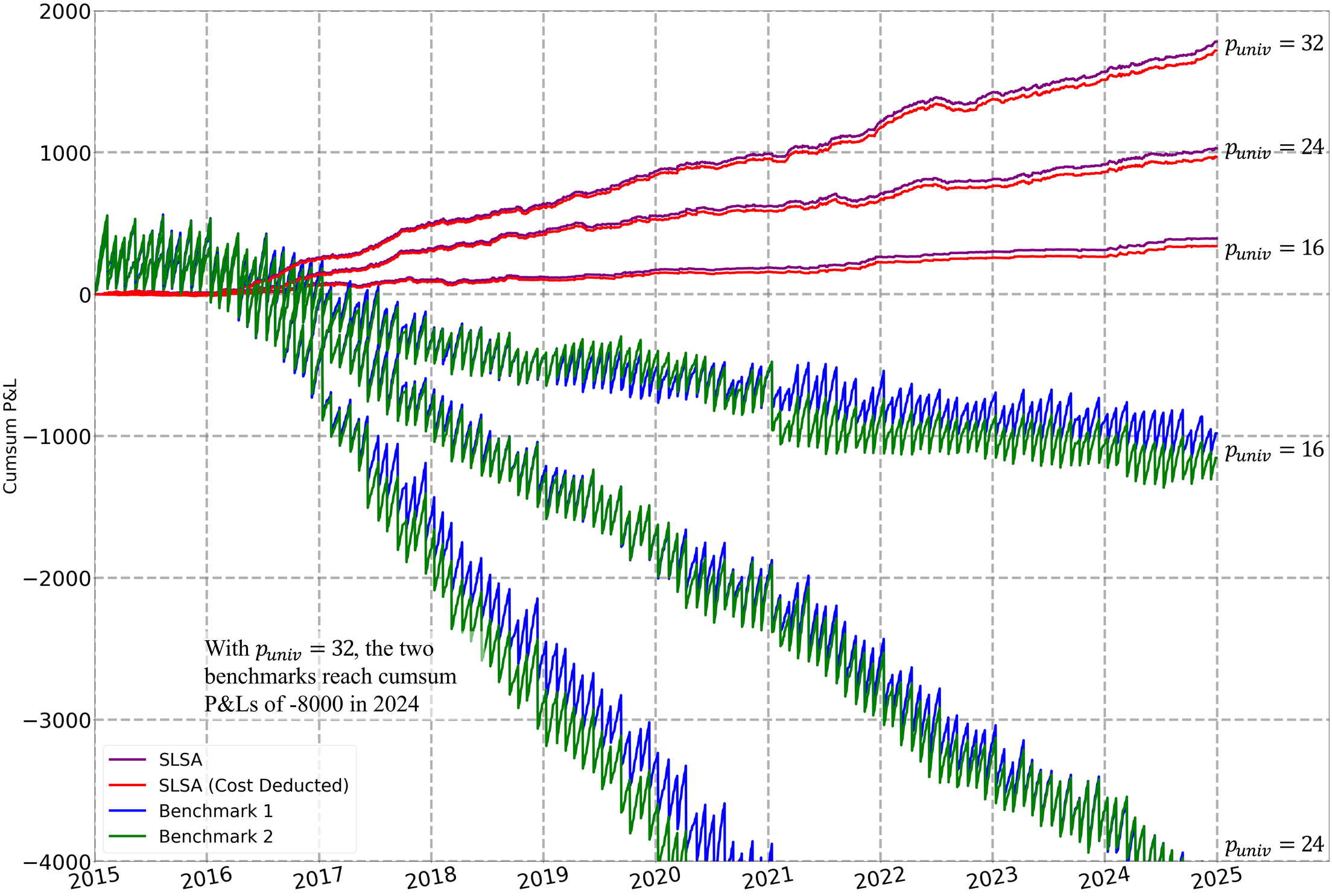}
    \caption{Comparison of Cumulative P\&Ls: Projected SLSA and Benchmark Strategies 
    }
    \label{fig:slsa_pnl}
    \label{fig:bm1_pnl}
    \label{fig:bm2_pnl}
\end{figure*}

\begin{table}[thbp]
\centering
\caption{Averages of Annual Metrics of Projected SLSA 
}
\label{tab:slsa_summary}
\begin{tabular}{lrrr}
\toprule
$p_{univ}$ & 16 & 24 & 32 \\
\midrule
Cumsum P\&L & 39.2621 & 102.8484 & 178.0427 \\
Information Ratio & 0.0938 & 0.1592 & 0.2351 \\
Max Drawdown & 14.4298 & 24.0106 & 23.1582 \\
Sortino Ratio & 0.1796 & 0.2585 & 0.4007 \\
Hit Rate & 0.4841 & 0.5327 & 0.5762 \\
HHI & 0.1033 & 0.0681 & 0.0549 \\
Effective N & 9.9559 & 15.0034 & 18.8209 \\
$avg \vert S-K\vert$ & 8.4020 & 10.3882 & 11.8789 \\
$avg \vert M-t\vert$ & 19.0769 & 23.4991 & 26.3175 \\
\bottomrule
\end{tabular}
\end{table}

\begin{figure*}[thbp] 
    \centering
    \includegraphics[width=1.0\textwidth]{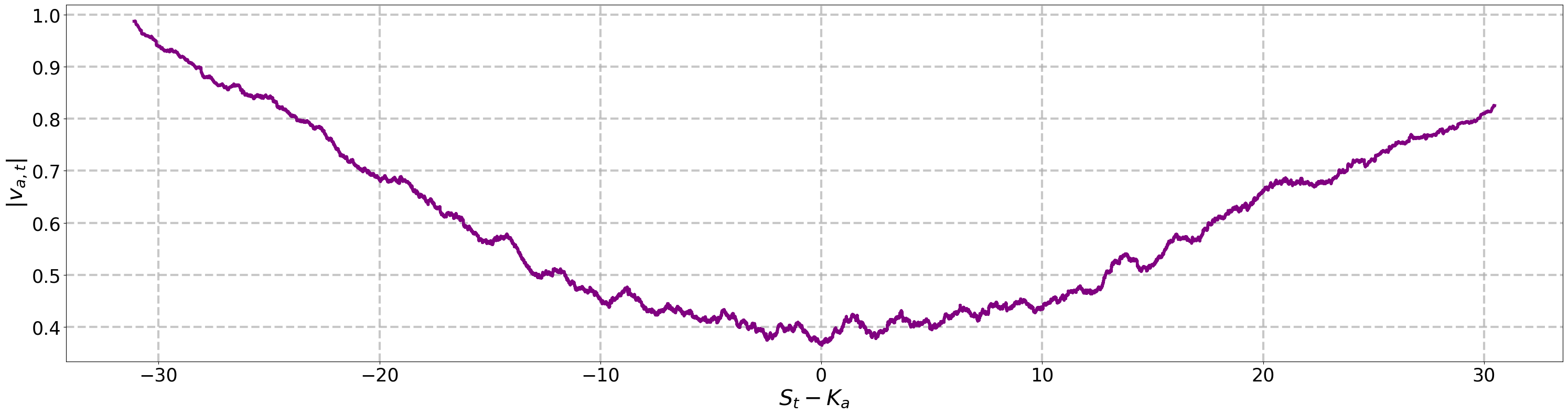}
    \caption{Rolling Average of $|v_{a,t}|$ over $S_t-K_a$}
    \label{fig:moneyness_arbitrages}
\end{figure*}

\begin{figure*}[thbp] 
    \centering
    \includegraphics[width=1.0\textwidth]{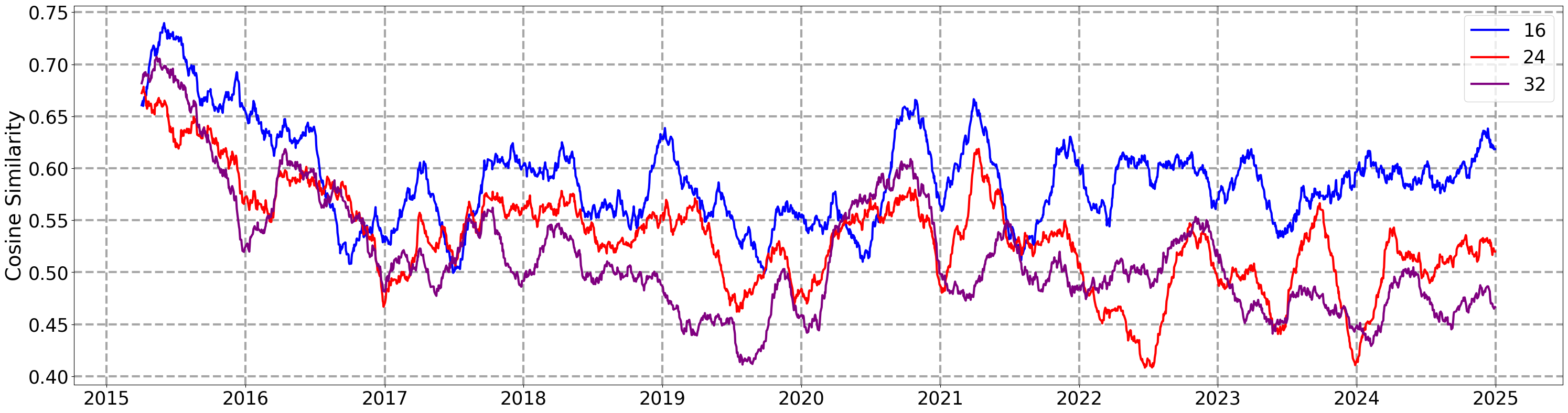}
    \caption{63-day Rolling Mean of Cosine Similarity between Prediction $\hat{\textbf{v}}_{t}$ and Position $\textbf{n}_t$}
    \label{fig:cosine_similarity}
\end{figure*}

\begin{table}[thbp]
\centering
\caption{Correlation between $K$ and $\textbf{n}_t$ 
}
\label{tab:k_position_corr}
\begin{tabular}{lrrr}
\toprule
$p_{univ}$ & Projected SLSA & Benchmark 1 & Benchmark 2 \\
\midrule
16 & 0.0000 & 0.1081 & 0.1085 \\
24 & 0.0000 & 0.1415 & 0.1413 \\
32 & 0.0000 & 0.1636 & 0.1645 \\
\bottomrule
\end{tabular}
\end{table}

Compared to the projected SLSAs, Benchmarks 1 and 2 exhibit steady downward slopes with spikes and drops around maturities. 
Due to the synthetic bond component in \eqref{eq:inception_payoff_bm2}, the benchmarks may generate larger profits on each trading date.  
However, they are also exposed to the cash flows required to clear these bonds at each maturity, as described in \eqref{eq:maturity_payoff_bm1} and \eqref{eq:maturity_payoff_bm2}.  
Since the positions generated by Benchmarks 1 and 2 are positively correlated with strike prices $K$, as shown in \tref{tab:k_position_corr}, the maturity cash flows (\eqref{eq:maturity_payoff_bm1} and \eqref{eq:maturity_payoff_bm2}) become negative in many cases, leading to the observed drops on maturities. 
From these results, we confirm the necessity of \eqref{eq:def_sa1} and \eqref{eq:def_sa2}.







\section{Conclusion}
\label{sec:conclusion}
This paper presents a two-step graph learning framework to exploit StatArbs in options markets, addressing two key gaps in the literature: \proman{1} the lack of direct deep learning methods for StatArbs in options markets, and \proman{2} overlooking the tabular form of the features commonly seen in practice.
In the first step, we present a prediction target that provably contains pure arbitrages. 
To predict this, we introduce RNConv, a novel architecture that integrates tree-based structures with graph convolution to effectively utilize tabular node features.
In the second step, we propose SLSA, a class of synthetic long positions theoretically shown to contain only arbitrage opportunities. 
We further prove that SLSA has the lowest risk level and possesses neutrality to all \cite{black1973pricing} risk factors under the arbitrage-free assumption. 
Moreover, we introduce the SLSA projection, which maps model predictions to SLSA positions capable of realizing StatArbs.
Experimental results show that RNConv statistically significantly outperforms the benchmarks, and SLSA positions derived from RNConv predictions yield consistent, upward-sloping cumulative P\&L curves, achieving an average information ratio of 0.1627 based on P\&L-contract returns.
Therefore, we offer a perspective on the prediction target and strategy for capitalizing on StatArbs in options markets within the context of deep learning in tandem with pioneering tree-based graph learning. 

Nevertheless, there remains room for future work. 
First, the underlying causes of the arbitrages within our prediction target need to be elucidated. 
Moreover, such an investigation may lead to more effective predictive features.  
Incorporating  edge features into our method may further enhance predictive power. 
Next, since fractional contracts cannot be traded in practice, methods for transforming them into integer-valued positions warrant further study. 
Lastly, depending on the investor’s circumstances (e.g., budget), more realistic constraints—such as limiting the investment to a smaller universe—may need to be considered, which is left for future work.




\bibliographystyle{informs2014} 
\bibliography{sections/08.references} 





  



\end{document}